\documentclass{amsart}

\usepackage{amsmath,amssymb}
\usepackage{amsthm}
\usepackage{bm}
\usepackage[alphabetic]{amsrefs}
\usepackage{indentfirst}
\usepackage{mathrsfs}
\usepackage{graphicx}
\usepackage{float}
\usepackage{placeins}
\usepackage{subcaption}
\usepackage{hyperref}
\usepackage{xcolor}
\usepackage{fourier}
\usepackage[margin=1.5in]{geometry}
\usepackage{soul}
\sethlcolor{yellow}

\usepackage{todonotes}

\theoremstyle{plain}
\newtheorem{theorem}{Theorem}
\newtheorem{lemma}{Lemma}

\newtheorem{prop}[theorem]{Proposition}
\newtheorem{coro}[theorem]{Corollary}

\theoremstyle{definition}
\newtheorem{defn}[theorem]{Definition}
\newtheorem*{notation}{Notation}

\theoremstyle{remark}
\newtheorem{remark}{Remark}

\DeclareMathOperator{\spanop}{span}

\newcommand{\thmref}[1]{Theorem~\ref{#1}}

\newcommand{\RR}{\mathbb{R}}

\newcommand{\CC}{\mathbb{C}}

\newcommand{\PP}{\mathbb{P}}
\newcommand{\EE}{\mathbb{E}} 
\DeclareMathOperator{\var}{Var}

\IfFileExists{mathabx.sty}%
  {\DeclareFontFamily{U}{mathx}{\hyphenchar\font45}%
   \DeclareFontShape{U}{mathx}{m}{n}{<->mathx10}{}%
   \DeclareSymbolFont{mathx}{U}{mathx}{m}{n}%
   \DeclareFontSubstitution{U}{mathx}{m}{n}%
   \DeclareMathAccent{\widebar}{0}{mathx}{"73}%
}{%
  \PackageWarning{mathabx}{%
    Package mathabx not available, therefore\MessageBreak substituting
    widebar with overline\MessageBreak }%
  \newcommand{\widebar}[1]{\overline{#1}}%
}

\newcommand{\eps}{\epsilon}

\title{Quantum Fisher information matrix via its classical counterpart from random measurements}

\thanks{This work is supported in part by National Science Foundation via award DMS-2309378. Kecen Sha thanks the Rhodes Information Initiative and the Department of Mathematics at Duke University for their hospitality where this work was carried out. The authors are grateful to Sisi Liu for bringing references \cites{MasahitoHayashi_1998,PhysRevLett.120.030404} to our attention and for helpful discussions. }

\author{Jianfeng Lu}
\address{Departments of Mathematics, Physics, and Chemistry, Duke University, Durham, NC 27708 USA}
\email{jianfeng@math.duke.edu}

\author{Kecen Sha}
\address{School of Mathematical Sciences, Peking University, Beijing, 100871 China}
\email{kecen.sha@stu.pku.edu.cn}
\date{\today}

\begin{document}

\begin{abstract}
Preconditioning with the quantum Fisher information matrix (QFIM) is a popular approach in quantum variational algorithms. Yet the QFIM is costly to obtain directly, usually requiring more state preparation than its classical counterpart: the classical Fisher information matrix (CFIM). It is known that averaging the classical Fisher information matrix over Haar-random measurement bases yields $\mathbb{E}_{U\sim\mu_H}[F^U(\boldsymbol{\theta})] = \frac{1}{2}Q(\boldsymbol{\theta})$ for pure states in $\mathbb{C}^N$. In this paper, we review this identity by revealing its connection to covariant measurement in quantum metrology. Furthermore, we go beyond this and obtain the exact variance of CFIM ($O(N^{-1})$), estimate its moment, and establish non-asymptotic concentration bounds ($\exp(-\Theta(N)t^2)$), demonstrating that using few random measurement bases is sufficient to approximate the QFIM accurately in high-dimensional settings.  This work establishes a solid theoretical foundation for efficient quantum natural gradient methods via randomized measurements. 
\end{abstract}
\maketitle

\section{Introduction}

Variational algorithms for quantum states have a long history and have also received renewed attention in recent years due to their application in quantum computing. In such algorithms, parameterized ansatz are used in a variational principle so that parameters are determined via optimization. Examples include variational Monte Carlo \cites{foulkes2001quantum, sorella2005wave, toulouse2016introduction}, where the ground state wave function is parameterized, and hybrid classical-quantum algorithms \cite{cerezo2021variational}, where quantum circuits are typically parameterized. In these methods, parameters in the ansatz are updated in the variational procedure typically through gradient-based optimization algorithms, such as stochastic gradient descent. 

To improve the performance of such optimization algorithms, preconditioners are commonly used. In particular, mimicking the popular natural gradient algorithms \cites{amari1998natural, amari2016information}, which incorporate geometric information about the parameter space, the quantum natural gradient algorithm has been proposed in \cite{stokes2020quantum} and widely used. It also has similarity to the stochastic reconfiguration method in the context of variational Monte Carlo \cites{sorella1998green, sorella2007weak}.  

In the quantum natural gradient method, we use the quantum Fisher information matrix (see Definition~\ref{def:QFIM}, not to be confused with the quantum Fisher information; see e.g., \cite{rath2021quantum} where random measurement protocol is considered) as a preconditioner to incorporate geometric information of the parameterized quantum states. In practical implementations, it is demanding to obtain the quantum Fisher information matrix (QFIM), and thus in recent works \cite{kolotouros2024randomnatural}, it was proposed to approximate the QFIM by their classical analogs, the classical Fisher information matrix (CFIM) corresponding to the probability amplitudes obtained from the wave function. The CFIM depends on the basis used to measure the quantum state, and thus the approximation is basis dependent. 

It was conjectured \cite{kolotouros2025accelerating} that the classical Fisher information matrix averaged over random choice of measurement basis would give the quantum Fisher information matrix, however, previous work \cite{kolotouros2024randomnatural} only provided numerical evidence that this relation might hold. {Albeit stated in somewhat different terminologies, this relation is in fact known in the literature of quantum metrology: The average classical Fisher information matrix is equivalent to the CFIM of covariant measurement, where the latter is known to be exactly half the QFIM} \cites{MasahitoHayashi_1998,PhysRevLett.120.030404}. 
{This equivalence will be further elaborated in Section} \ref{sec:relation}. {We also provide a direct proof for self-containedness, which also yields slightly stronger statements. Moreover, we go beyond expectation results to further provide concentration results, and thus enable quantitative bounds of estimating QFIM using a finite number of measurement bases.}

The remainder of the paper is organized as follows. We will introduce the setup and state the main results in Section~\ref{sec:main}. We also provide some numerical experiments validating the concentration bounds. The other sections are devoted to the proof of the results. Section~\ref{sec:relation} {builds the connection between the expectation of CFIMs under random measurement basis and the CFIM of covariant measurement.} Section~\ref{sec:mean} gives an alternative proof which will be helpful to derive further results. We quantify the variance and establish moment estimation of the random CFIMs in Section~\ref{sec:variance} which quantifies the fluctuations. In Section~\ref{sec:concentration}, we establish concentration bounds of CFIM around its mean (half the QFIM). We summarize by some remarks and future directions in Section~\ref{sec:conclusion}.

\section{Setup and main results}\label{sec:main}

In this section, we will first recall the definitions, and then state our results for the classical Fisher information matrix under randomly sampled measurement basis, including the expectation, the variance, and the concentration bounds. Additionally, we will examine the tightness of our concentration bounds via numerical experiments.

Let us commence with the formal definition of the information matrices to be analyzed, note that the definitions might be different up to a constant in different papers. 
Throughout, we will consider $\psi_{\bm\theta}$ as a family of pure quantum states in $\CC^N$ that is parameterized by $\bm\theta\in\RR^m$, \textit{i.e.}, for all $\bm\theta$, $\psi_{\bm\theta}$ is normalized. 
We will assume $N \geq 2$ and that $N$ is finite.

\begin{defn}
\label{def:QFIM}
The Quantum Geometry Tensor (QGT) at $\bm\theta$, denoted as $\mathcal{Q}(\bm\theta)\in\mathbb{C}^{m\times m}$, is defined as
    \begin{equation}
        \label{eq:QGT}
        \mathcal{Q}_{ij}(\bm\theta)=\left\langle\frac{\partial\psi_{\bm\theta}}{\partial\theta_i},\frac{\partial\psi_{\bm\theta}}{\partial\theta_j}\right\rangle_{\mathbb{C}}-\left\langle\frac{\partial\psi_{\bm\theta}}{\partial\theta_i},\psi_{\bm\theta}\right\rangle_{\mathbb{C}}\left\langle\psi_{\bm\theta},\frac{\partial\psi_{\bm\theta}}{\partial\theta_j}\right\rangle_{\mathbb{C}}.
    \end{equation}
    The quantum Fisher information matrix (QFIM) at $\bm\theta$, denoted as $Q(\bm\theta)=\operatorname{Re}(\mathcal{Q}(\bm\theta))$, is the real part of the QGT. Geometrically, the QFIM corresponds to the pullback of the Fubini-Study metric from the projective Hilbert space to the parameter manifold. Here and in the sequel, $\langle \cdot, \cdot \rangle_{\mathbb{C}}$ denotes the standard complex inner product on $\mathbb{C}^N$.
\end{defn}
\begin{defn}
    Given a measurement basis $U=[\bm{u}_1,\cdots,\bm{u}_N]\in \mathrm U(N)$, denote $\bm{p}^U(\bm\theta)\in\RR^N$ as the probability distribution on $N$ elements: $[\bm p^U(\bm\theta)]_i=|[U^*\psi_{\bm\theta}]_i|^2$. The classical Fisher information matrix (CFIM) under basis $U$, denoted as $F^U(\bm\theta)$, is defined as:
    \begin{equation}
    \label{CFIM:def}
    F^U_{ij}(\bm\theta)=\frac{1}{4}\EE_{\bm{p}^U_{\bm\theta}}\left[\left(\frac{\partial \log{\bm{p}^U_{\bm\theta}}}{\partial\theta_i}\right)^\top\frac{\partial \log{\bm{p}^U_{\bm\theta}}}{\partial\theta_j}\right]=\left\langle\frac{\partial \sqrt{\bm{p}^U_{\bm\theta}}}{\partial\theta_i},\frac{\partial \sqrt{\bm{p}^U_{\bm\theta}}}{\partial\theta_j}\right\rangle,
\end{equation}
where the operations $\sqrt{\bm{p}^U_{\bm\theta}}$ and $\log{\bm{p}^U_{\bm\theta}}$ are applied element-wise to the vector $\bm{p}^U_{\bm\theta}$. 
\end{defn}

The following result states the connection between classical and quantum Fisher information matrices. 
\begin{theorem}
\label{1}
     If the measurement basis $U$ is drawn from the Haar distribution $\mu_H$ on $\mathrm U(N)$, then the average CFIM satisfies
    \begin{equation}
        \EE_{U\sim\mu_H}[F^U(\bm\theta)]=\frac{1}{2}\operatorname{Re}(\mathcal{Q}(\bm\theta))=\frac{1}{2}Q(\bm\theta).
    \end{equation}
\end{theorem}
\thmref{1} implies that by measuring the quantum state under random bases, one can approximate the QFIM by the average CFIM. Thus, the geometry of quantum states can be characterized by sampled classical Fisher matrices. 

The following theorem gives the variance of the random matrix, which can be used to quantify the approximation, say by the central limit theorem.
\begin{theorem}
    \label{2}
     If the measurement basis $U$ is drawn from the Haar distribution $\mu_H$ on $\mathrm U(N)$, then the variance of the random CFIM satisfies 
     \begin{equation}
         \var_{U\sim\mu_H}[F^U(\bm\theta)]=\frac{1}{8N}\left(\operatorname{diag}(\mathcal{Q}(\bm\theta))\operatorname{diag}(\mathcal{Q}(\bm\theta))^\top+\mathcal{Q}(\bm\theta)\odot\mathcal{Q}(\bm\theta)^\top\right),
     \end{equation}
     where the variance $\var_{U\sim\mu_H}[F^U(\bm\theta)]$ is computed  element-wise for the random matrix $F^U(\bm\theta)$, $\operatorname{diag}(\mathcal{Q}(\bm\theta))$ denotes the column vector of the main diagonal entries of $\mathcal{Q}(\bm\theta)$ and  $\odot$ denotes the Hadamard product (i.e., entrywise product) of two matrices.
\end{theorem}
\thmref{2} shows that each entry of the random CFIM has variance of order  $O(N^{-1})$. Since $N=2^n$, where $n$ is the number of qubits, the approximation accuracy improves exponentially with $n$. The variance depends not only on the QFIM but also the imaginary part of the QGT, which reflects geometric phase information of a quantum system.

   The next theorem establishes a rigorous characterization of the concentration behavior of each entry of the random CFIM by providing a two-sided estimation of its moment.

\begin{theorem}[Moment estimation]
\label{thm:6}
Assume that $Q_{ii}(\bm\theta)>0$ for all $1\leq i\leq m$, denote random variables $X_{ij}=\frac{F^U_{ij}(\bm\theta)-\frac{1}{2}Q_{ij}(\bm\theta)}{\sqrt{Q_{ii}(\bm\theta)Q_{jj}(\bm\theta)}}$. Then there exist absolute constants $c_1,c_2$, such that for all $1\leq i,j\leq m,k\in \mathbb N^*,N\geq k$, we have
    \begin{equation}
    \label{eq:momentsk}
        \EE[X_{ij}^{2k+1}]=0,\quad \frac{c_1\sqrt k}{\sqrt N}\leq \bigl(\EE[X_{ij}^{2k}]\bigr)^{\frac{1}{2k}}\leq \frac{c_2\sqrt k}{\sqrt N}.
    \end{equation}
Moreover, there exist absolute constants $c_3,c_4>0$ such that for all $1\leq i,j\leq m$,
\begin{equation}
\label{eq:upper-expmom}
    \EE[\exp({\lambda X_{ij}})]\leq \exp (c_3\lambda^2/N),\quad \forall\lambda \in\mathbb R, 
\end{equation}
and
\begin{equation}
\label{eq:lower-expmom}
    \EE[\exp({\lambda X_{ij}})]\geq \exp (c_4\lambda^2/N),\quad \forall|\lambda|\leq  N. 
\end{equation}
\end{theorem}
\begin{remark}
    If $Q_{ii}(\bm\theta)$ or $Q_{jj}(\bm\theta)=0$, since $F^U(\bm\theta)\preceq Q(\bm\theta)$, we must have $F^U_{ij}(\bm\theta)=0$. This is a degenerate case, and the structure is trivial. For convenience, in the rest of the paper, we always assume that $Q_{ii}(\bm\theta)>0$. The degenerate case does not influence our results. 
\end{remark}
 The exponential moment upper bound in \thmref{thm:6} proves that $F^U_{ij}(\bm\theta)$ is a sub-Gaussian random variable with a proxy variance of $O(1/N)$. This ensures that the probability of $F^U_{ij}(\bm\theta)$ deviating from its mean decays at least as fast as $\exp(-cNt^2)$.

 In addition, the lower bound is crucial as it demonstrates that the sub-Gaussian rate is asymptotically tight as $N\to+\infty $. By providing a matching lower bound for the exponential moment within the local regime $|\lambda| \le N$, \thmref{thm:6} confirms that the $\exp(-cNt^2)$ tail bound is the intrinsic physical limit of the CFIM's variability.

As a supplement, we give concentration bounds for the random CFIM in terms of standard matrix norms.  Moreover, eigenvalues of the error matrix $F^U(\bm\theta)-\EE[F^U(\bm\theta)]$ can be uniformly controlled with high probability.

\begin{theorem}[Maximum norm]
\label{3}
    If the measurement basis $U$ is drawn from the Haar distribution $\mu_H$ on $\mathrm U(N)$, then for every $t>0$, we have
    \begin{equation}
        \PP\left(\frac{\|F^U(\bm\theta)-\EE[F^U(\bm\theta)]\|_{\max}}{\|\EE[F^U(\bm\theta)]\|_{\max}}\geq t\right)\leq 2m^2\exp\left(-\frac{(N-1)t^2}{120}\right),
    \end{equation}
    where $\|A\|_{\max} \triangleq\max\limits_{1\leq i,j\leq m}|A_{ij}|$ is the maximum entrywise norm of the matrix $A$.
\end{theorem}
\begin{theorem}[Frobenius norm]
\label{4}
       If the measurement basis $U$ is drawn from the Haar distribution $\mu_H$ on $\mathrm U(N)$, then for every $t>0$, we have 
          \begin{equation}
        \PP\left(\frac{\|F^U(\bm\theta)-\EE[F^U(\bm\theta)]\|_F}{\|\EE[F^U(\bm\theta)]\|_{F}}\geq t+\sqrt\frac{2m}{N}\right)\leq \exp\left(-\frac{(N-1)t^2}{120}\right),
    \end{equation}
where $\|A\|_{F} \triangleq \sqrt{\sum\limits_{i,j=1}^m|A_{ij}|^2}$ is the Frobenius norm of the matrix $A$.    
\end{theorem}

\begin{theorem}[Eigenvalue control]
\label{5}
    Let $\varepsilon\in(0,\frac{1}{2})$. If the measurement basis $U$ is drawn from the Haar distribution $\mu_H$ on $\mathrm U(N)$ and  $N\geq\frac{10^5m}{\varepsilon^2}$, then we have
    \begin{equation}
        \PP\left((1-2\varepsilon)\EE[F^U(\bm\theta)]\preceq F^U(\bm\theta)\preceq (1+2\varepsilon)\EE[F^U(\bm\theta)]\right) \geq 1-\exp\left(-\frac{(\varepsilon\sqrt{N-1}-285\sqrt{m})^2}{30}\right).
    \end{equation}
\end{theorem}
\thmref{3}, \ref{4} and \ref{5} show that the CFIM concentrates around its mean when the dimension $N$ increases. In particular, when $N\gg m$, $F^U(\bm\theta)$ can be two-sided controlled by the QFIM with a slight scalar relaxation with high probability. This means that in terms of using as a preconditioner, the effectiveness of using QFIM or a single realization of CFIM is essentially the same as the CFIM is a high-quality spectral approximation of the QFIM with high probability.

We now discuss the tightness of \thmref{3}, \ref{4} and \ref{5}.  As we have shown in \thmref{thm:6},  $F^U_{ij}(\bm\theta)$ is sub-Gaussian with its $\|\cdot\|_{\psi_2}\sim \Theta(N^{-1/2})$, which supports the tightness of tail bound $\exp(-cNt^2)$ -- and the order is optimal in asymptotic regime. Note that the matrix norm is inherently bounded below by its entry-wise magnitudes, so concentration rate exceeding $\exp(-cNt^2)$ is almost impossible without additional structural assumptions on $Q(\bm\theta)$. For the tightness of the estimation of the expected Frobenius or spectral error norm, please refer to the remark after the proof (Remark \ref{remark:thm4} and Remark \ref{remark:thm5}).
We also provide numerical experiments to support our results.

\subsection*{Numerical experiments} We take the relative error $\frac{\|F^U(\bm\theta)-\EE[F^U(\bm\theta)]\|_F}{\|\EE[F^U(\bm\theta)]\|_{F}}$  in Theorem \ref{4} as an example to examine the tightness of the concentration inequality. We aim to show that the tail bound $\exp(-c N t^2)$ cannot be improved up to a constant factor $c$, and that the expected upper bound is of order $O(1/\sqrt{N})$. All experiments are conducted with $m = 10$.

\begin{figure}[htb]
    \centering
    \includegraphics[width=0.8\textwidth]{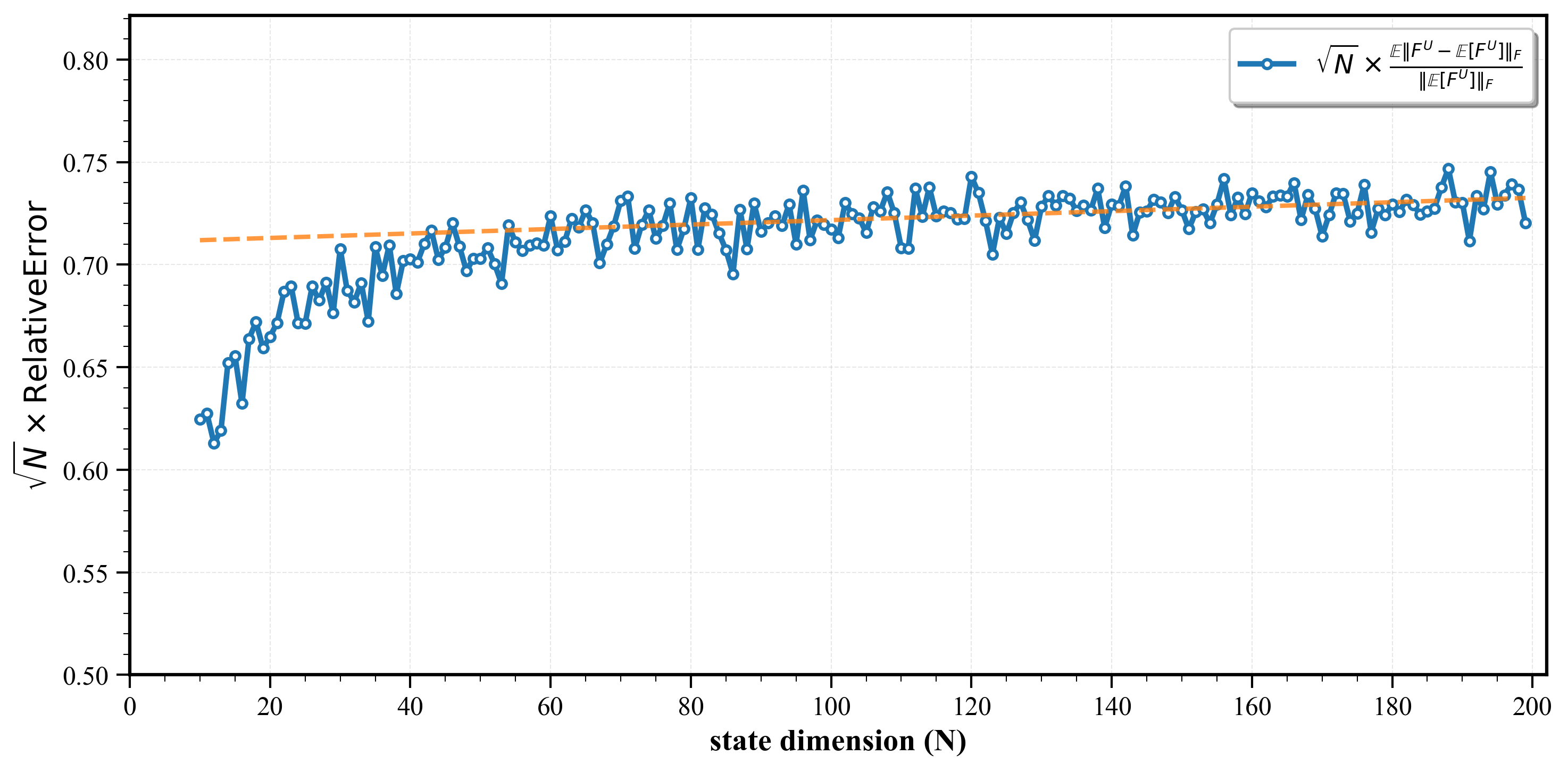}
    \caption{Relative Frobenius norm errors of different dimensions $N$. The linear fit (orange dashed line) is performed over the range $N \geq 60$ to better reflect this asymptotic behavior.}
    \label{fig:Frobenius_norm}
\end{figure}
Figure \ref{fig:Frobenius_norm} shows the plot of $\sqrt{N}$ multiplied by the average relative error against $N$. For each $N$, the result is averaged over 100 trials. The plot shows that the scaled error remains approximately constant as $N$ increases. While some non-linear fluctuations are visible at small values of $N$ due to finite-size effects, the trend stabilizes as $N$ grows. This indicates that the expected relative error scales as $O(1/\sqrt{N})$ in the asymptotic regime, consistent with the variance derived in \thmref{2}.

\begin{figure}[htb]
    \centering
    \includegraphics[width=0.8\textwidth]{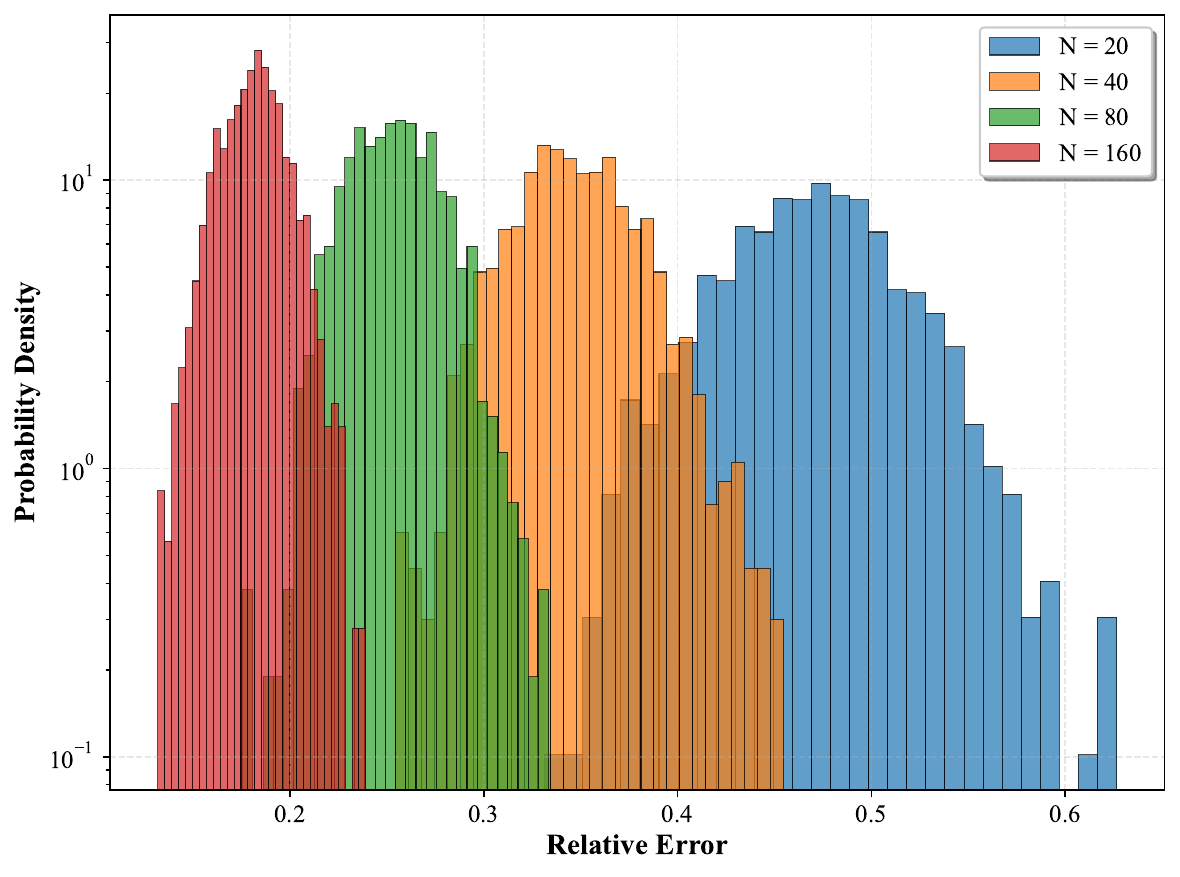}
    \caption{Histograms of the QFIM estimation error.}
    \label{fig:hisgram}
\end{figure}
Figure \ref{fig:hisgram} displays the frequency histograms of the relative error for $N = 20, 40, 80, 160$, each with a sample size of 1000. We observe that the error distribution becomes more concentrated around zero as \(N\) increases, for example, by comparing the error distribution for $N = 20$ versus $N = 160$. 
This visualizes the concentration phenomenon: As the dimension \(N\) increases, the CFIM becomes a more reliable estimator of the QFIM. The rapid decay of the histogram away from the center aligns with the exponential tail bounds proved in Theorems~\ref{3}-\ref{5}.

\begin{figure}[htb]
    \centering
    \begin{subfigure}{0.4\textwidth}
        \includegraphics[width=\textwidth]{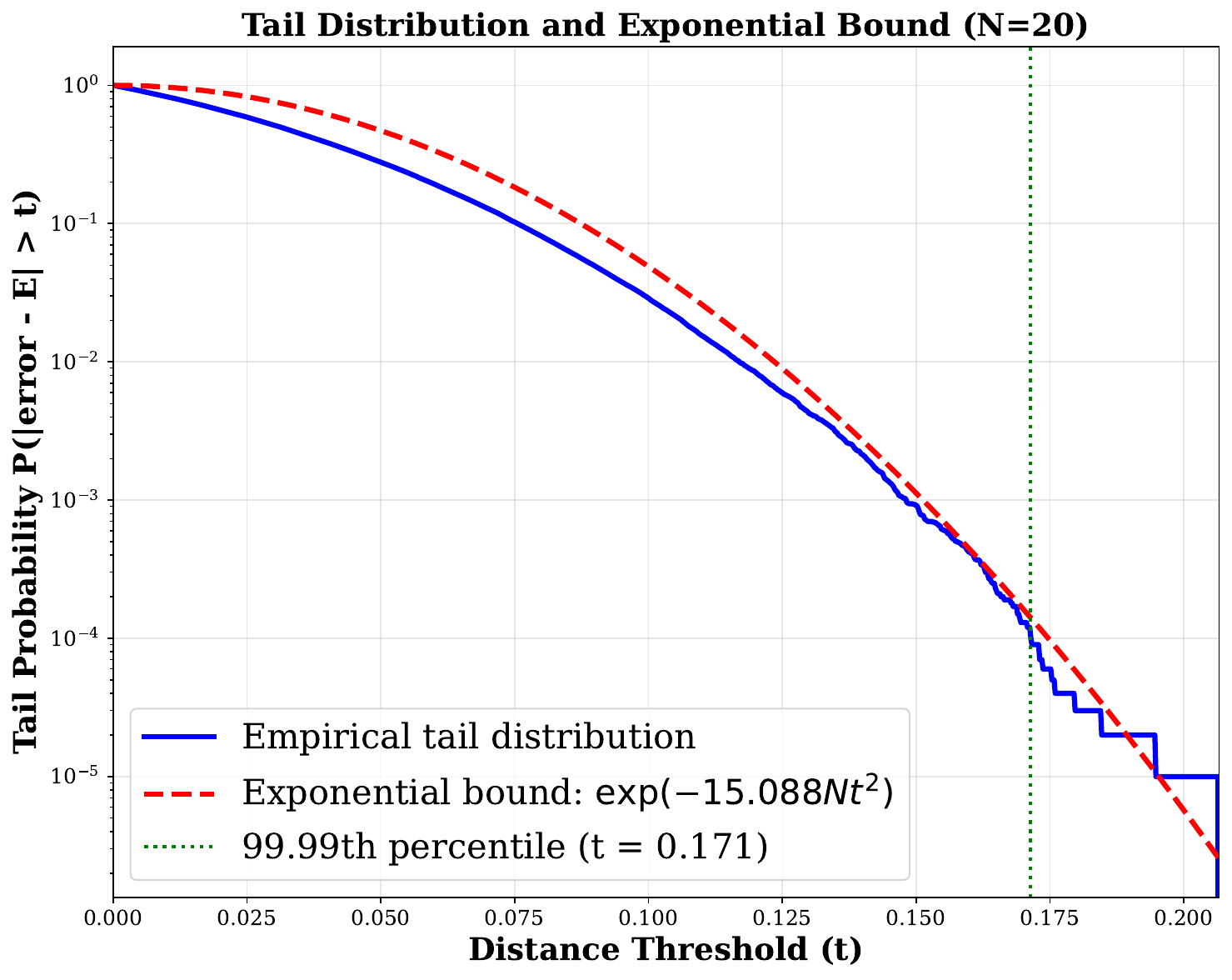}
        \caption{Empirical distribution for $N=20$}
        \label{fig:dist_n20}
    \end{subfigure}
        \qquad
    \begin{subfigure}{0.4\textwidth}
        \includegraphics[width=\textwidth]{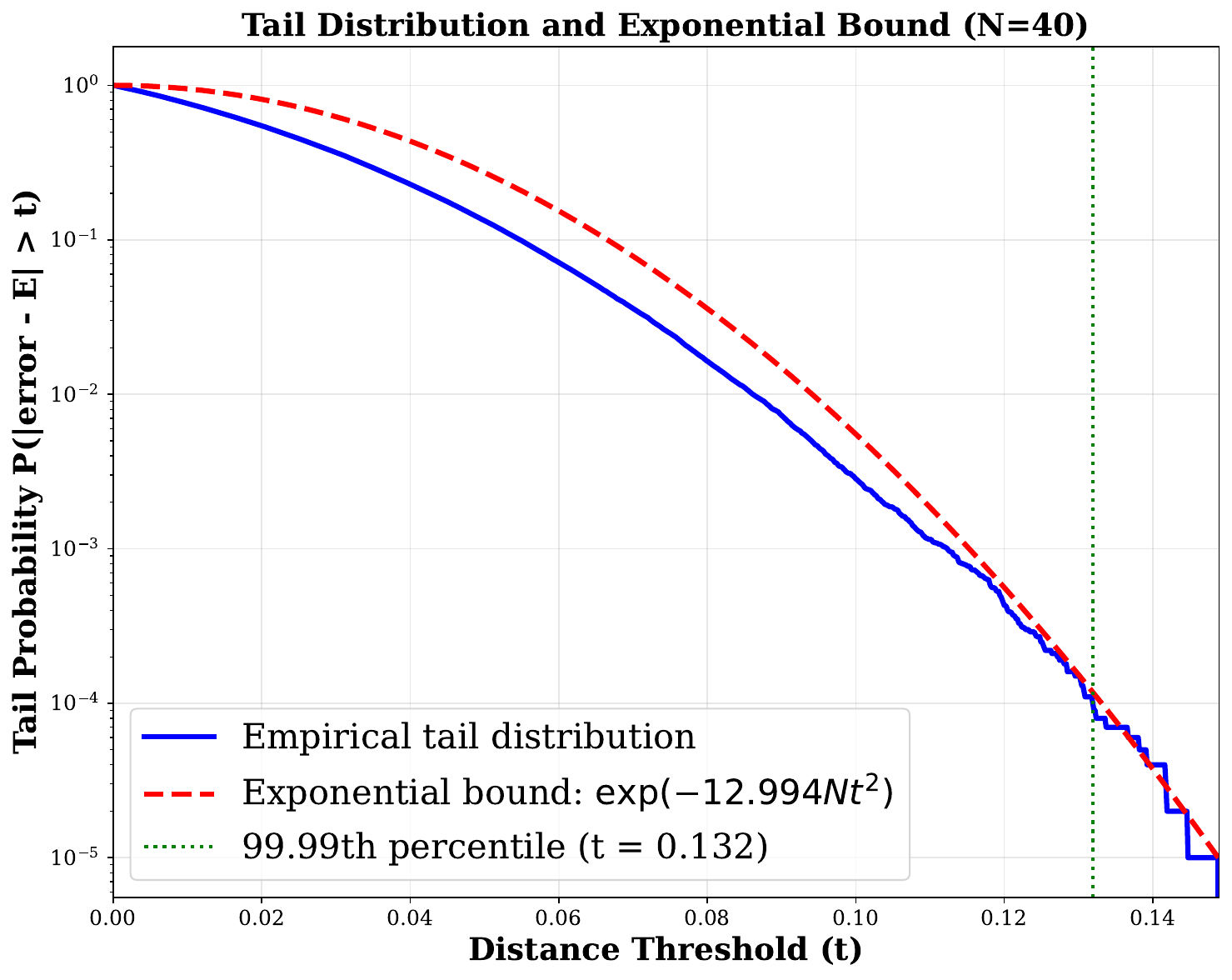}
        \caption{Empirical distribution for $N=40$}
        \label{fig:dist_n40}
    \end{subfigure}
    \qquad
    \begin{subfigure}{0.4\textwidth}
        \includegraphics[width=\textwidth]{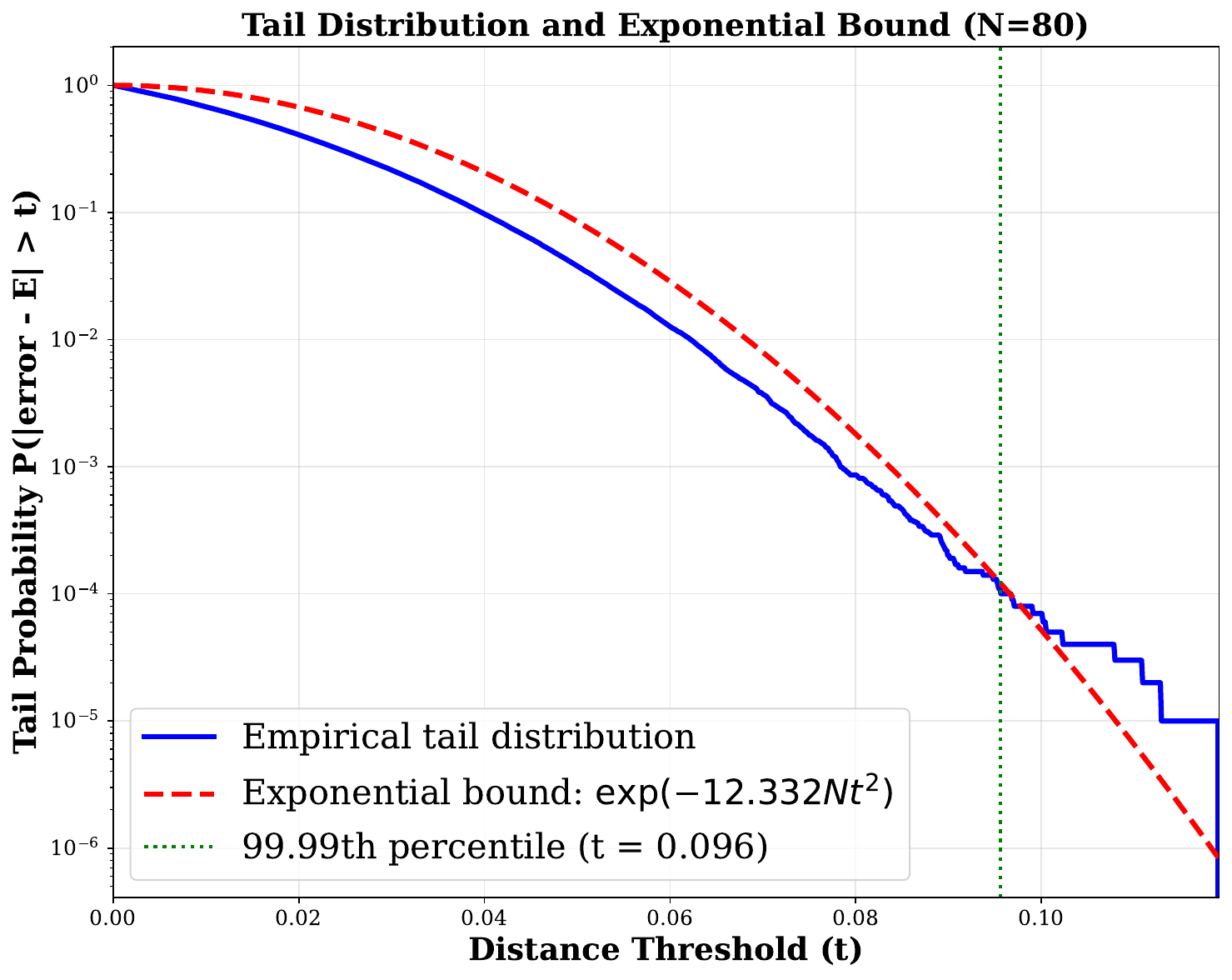}
        \caption{Empirical distribution for $N=80$}
        \label{fig:dist_n80}
    \end{subfigure}
    \qquad
        \begin{subfigure}{0.4\textwidth}
        \includegraphics[width=\textwidth]{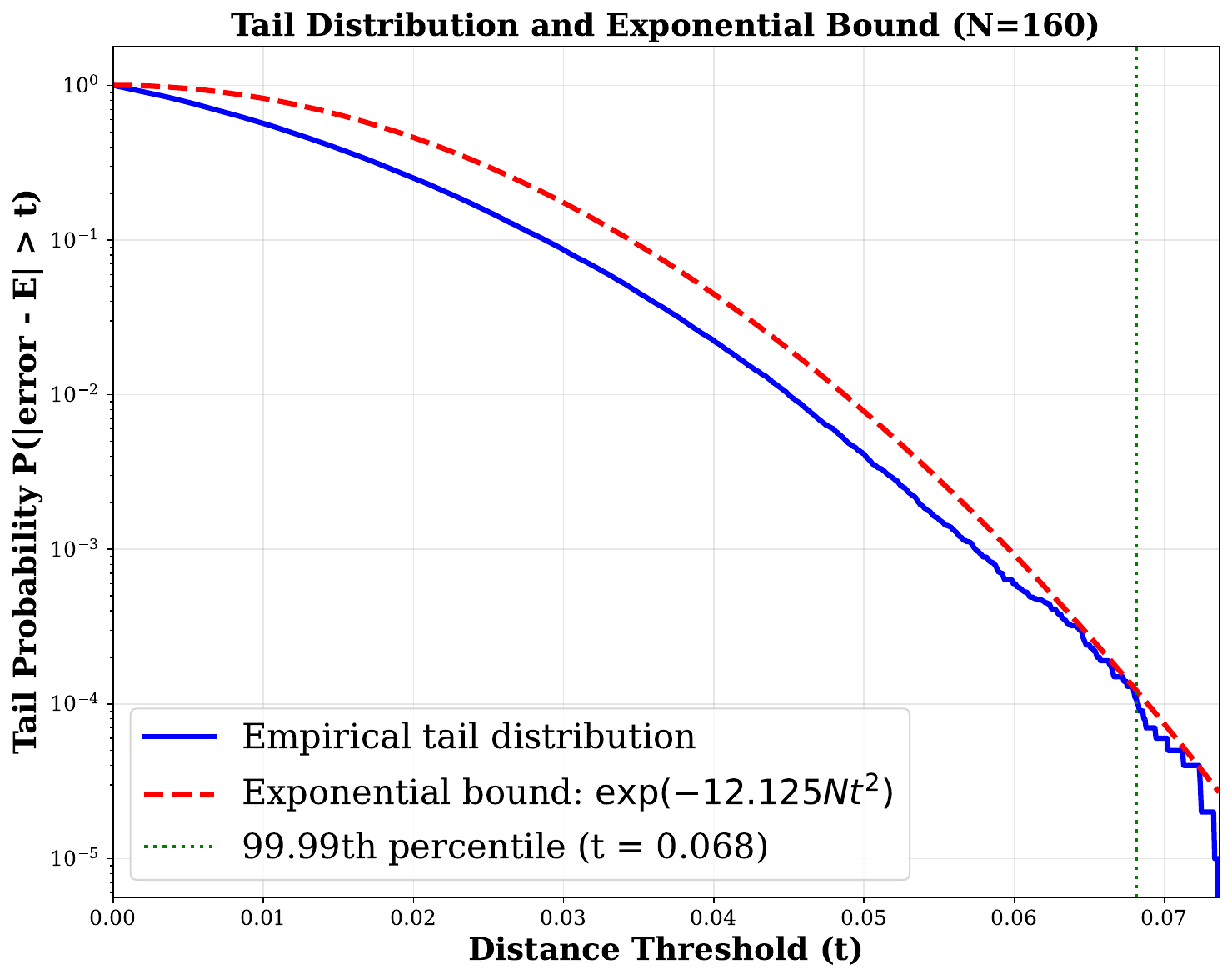}
        \caption{Empirical distribution for $N=160$}
        \label{fig:dist_n160}
    \end{subfigure}
    \caption{Empirical distribution functions and exponential upper-bound curves \(\exp(-cNt^2)\) for different values of $N$. For each $N$, the constant $c$ is the largest estimated value such that the empirical tail distribution (up to the $99.99$-th percentile, the green dashed line) lies entirely below the theoretical tail distribution curve. The specific estimation procedure is given in  Figure \ref{fig:anaylsis}.}
    \label{fig:distributions}
\end{figure}
Figure \ref{fig:distributions} presents the empirical tail distribution functions estimated from the frequency histograms for $N =20, 40, 80, 160$, based on 100,000 samples. The close agreement between the empirical tail and the theoretical curve indicates that the tail distribution roughly follows the form $\exp(-cNt^2)$. The results also show that the value of $c$ remains relatively stable as $N$ increases, supporting the view that the tail bound $\exp(-c N t^2)$ is tight up to a constant factor. This provides empirical evidence that the theoretical bounds in Theorem \ref{4} cannot be significantly improved.
We provide details of the estimation procedure in Figure \ref{fig:anaylsis}.

\begin{figure}[htb]
    \centering
    \includegraphics[width=0.8\textwidth]{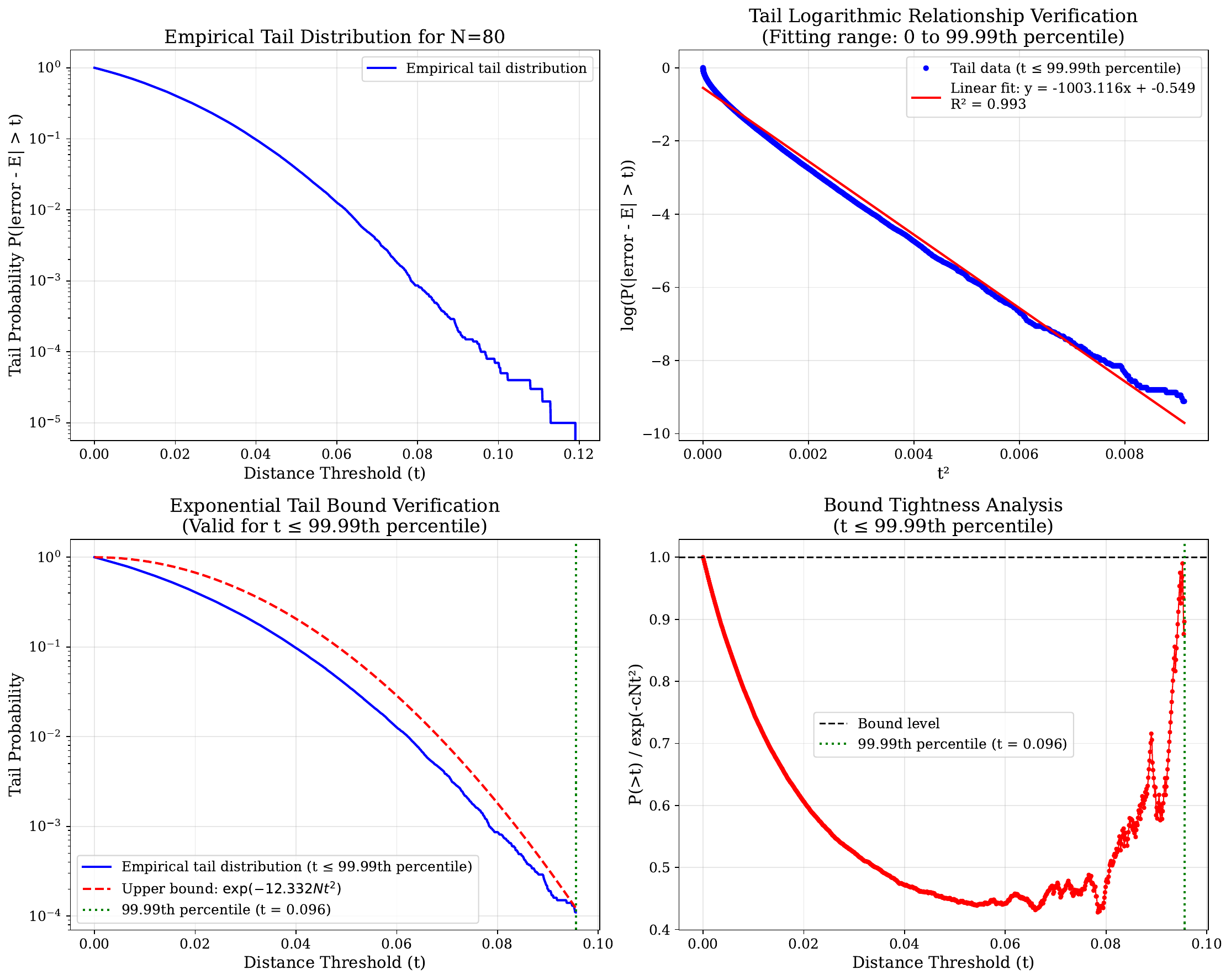}
    \caption{Tail distribution analysis for $N=80$. The top-left figure shows the empirical tail probability (complementary CDF) of the relative Frobenius norm error. The top-right figure plots the logarithm of the tail probability against $t^2$ and provides linear regression. Then we find the best $c$ based on the slope of the regression. The bottom-left figure depicts the best theoretical curve and the bottom-right figure plots the ratio between the empirical and the upper-bound. }
    \label{fig:anaylsis}
\end{figure}

\section{Connection with known results in quantum metrology}
\label{sec:relation}

We first review the basic terminology in quantum metrology to state the known result in the literature.
Let $\Omega$ be a measurable space, and $\mathbb H$ be a Hilbert space. Let $\Sigma$ be the set of all measurable sets in $\Omega$. A positive operator-valued measure (POVM) is a map $M:\Sigma\to \mathcal L(\mathbb H)$ such that
\[
\begin{aligned}
&M(B) = M(B)^*, \quad M(B) \succeq 0, \quad M(\varnothing) = 0, \quad M(\Omega) = \operatorname{Id},\quad \forall B\in\Sigma; \\
&M\left(\bigcup_{i} B_i\right) = \sum_i M(B_i) \quad \text{for } B_i \cap B_j = \varnothing \ (i \neq j),
\{B_i\} \text{ are countable subsets of } \Omega.
\end{aligned}
\]

Denote $\mathcal M(\Omega,\mathbb H)$ as the set of all measurements on $\mathbb H$ whose measurable set is $\Omega$. A measurement $M$ and a density matrix $\rho$ together define a probability distribution on $\Omega$, where the probability of the outcome belonging to any $B$ is $\operatorname{tr}(\rho M(B))$. For pure state $\psi\in\mathbb P\mathbb H$, $\rho=|\psi\rangle\langle\psi|$.

Let $G$ be a compact transitive Lie group of all transformations on a compact set $\Theta$ and $\{V_g\}$ is an irreducible representation of $G$ in $\mathbb H$. The covariant measurement is defined as follows:
\begin{defn}
    A measurement $\Pi\in\mathcal M(\Theta,\mathbb H)$ is said to be covariant with respect to $G$ if for any $g\in G$ and Borel set $B\subset\Theta$, we have
    \begin{equation}
        V_g^*\Pi(B)V_g=\Pi(B_{g^{-1}}),
    \end{equation}
    where $B_g\triangleq\{g\theta:\theta\in B\}$. Denote $\mathcal M(\Theta,\mathbb H,G)$ as the set of all covariant measurements.
\end{defn}
Denote $\nu$ as the invariant measure on $\mathbb P\mathbb H$. Then $\mathcal M(\mathbb P\mathbb H,\CC^N,\mathrm U(N))$ contains a unique covariant measurement of the following form:
\begin{equation}
\label{def:covariant}
\Pi(B)\triangleq \int_B N|\phi\rangle\langle\phi| \nu(d\phi).
\end{equation}
This measurement acting on the pure state set $\mathbb P\mathbb H$ will give a family of distributions:
\[\{P_\psi\triangleq \operatorname{tr}(|\psi\rangle\langle\psi|\Pi(d\phi))= N|\langle\psi,\phi\rangle_{\CC}|^2\nu (d\phi):\psi\in\mathbb P\mathbb H\}.\] 
Denote $d_\Pi(\psi_1,\psi_2)$ as the geodesic distance between $P_{\psi_1},P_{\psi_2}$ with respect to the Fisher information metric. Lemma 2 in \cite{MasahitoHayashi_1998} states that:
\begin{lemma}
\label{lem in Hayashi}
    Let $d_{fs}$ be the geodesic distance of Fubini-Study metric. Then we have $d_\Pi=\sqrt{2}d_{fs}$.
\end{lemma}
We will show that Theorem \ref{1} can be implied by 
Lemma \ref{lem in Hayashi}. Suppose that the pure states are parameterized by $\bm\theta\in\RR^m$, then it is known that (see \cite{kolotouros2024randomnatural} for instance) $d^2_{fs}(\psi_{\bm\theta_0},\psi_{\bm\theta_0+\epsilon \bm\theta})=\frac{\eps^2}{2}\bm\theta^\top Q(\bm\theta_0)\bm\theta+o(\epsilon^2)$ and $d^2_\Pi(\psi_{\bm\theta_0},\psi_{\bm\theta_0+\epsilon \bm\theta})=4\cdot \frac{\epsilon^2}{2}\bm\theta^\top I(\bm\theta_0)\bm\theta+o(\epsilon^2)$, where $I(\bm\theta)$ is the Fisher information matrix for parameterized distribution $P_{\psi_{\bm\theta}}$, defined as
\begin{equation}
\label{covaraint information matrix}
   I_{ij}(\bm\theta)=\int_{\mathbb P\mathbb H}\frac{1}{4}\cdot \frac{\partial\log N|\langle\psi_{\bm\theta},\phi\rangle_{\CC}|^2}{\partial \bm\theta_i}\cdot \frac{\partial\log N|\langle\psi_{\bm\theta},\phi\rangle_{\CC}|^2}{\partial \bm\theta_j}\cdot N|\langle\psi_{\bm\theta},\phi\rangle_{\CC}|^2\nu (d\phi). 
\end{equation}
When regarding $\psi_{\bm\theta},\phi$ as a complex vector in $\CC\mathbb S^{N-1}$,  (\ref{covaraint information matrix}) is equivalent to
\[  I_{ij}(\bm\theta)=\int_{\mathbb C\mathbb S^{N-1}}\frac{1}{4}\cdot \frac{\partial\log N|\langle\psi_{\bm\theta},\phi\rangle_{\CC}|^2}{\partial \bm\theta_i}\cdot \frac{\partial\log N|\langle\psi_{\bm\theta},\phi\rangle_{\CC}|^2}{\partial \bm\theta_j}\cdot N|\langle\psi_{\bm\theta},\phi\rangle_{\CC}|^2\mu (d\phi), \]
where $\mu$ follows the uniform distribution on $\CC\mathbb S^{N-1}$.
By Lemma \ref{lem in Hayashi} we have $I(\bm\theta)=\frac{1}{2}Q(\bm\theta)$. Then
 the average $\EE_{U\sim\mu_H}[F_{ij}^U(\bm\theta)]$ can be written as 
\begin{align*}
    \EE_{U\sim\mu_H}[F_{ij}^U(\bm\theta)]&=\EE_{U\sim\mu_H}\left[\sum\limits_{k=1}^N\frac{\partial \sqrt{(\bm{p}^U_{\bm\theta})_k}}{\partial\theta_i}\cdot \frac{\partial \sqrt{(\bm{p}^U_{\bm\theta})_k}}{\partial\theta_j}\right]\\
    &=\sum\limits_{k=1}^N\EE_{U\sim\mu_H}\left[\frac{1}{4(p_{\bm\theta}^U)_k}\cdot \frac{\partial {(\bm{p}^U_{\bm\theta})_k}}{\partial\theta_i}\cdot \frac{\partial {(\bm{p}^U_{\bm\theta})_k}}{\partial\theta_j}\right]\\
    &=\sum\limits_{k=1}^N\EE_{U\sim\mu_H}\left[\frac{1}{4|\bm e_k^* U^*\psi_{\bm\theta}|^2}\cdot \frac{\partial {|\bm e_k^* U^*\psi_{\bm\theta}|^2}}{\partial\theta_i}\cdot \frac{\partial {|\bm e_k^* U^*\psi_{\bm\theta}|^2}}{\partial\theta_j}\right]\\
    &=\sum\limits_{k=1}^N\EE_{\phi_k\sim\mu}\left[\frac{1}{4N^2| \phi_k^*\psi_{\bm\theta}|^2}\cdot \frac{\partial {N|\phi_k^*\psi_{\bm\theta}|^2}}{\partial\theta_i}\cdot \frac{\partial {N|  \phi_k^*\psi_{\bm\theta}|^2}}{\partial\theta_j}\right]\\
    &=N\cdot\frac{1}{N} I_{ij}(\bm\theta)=\frac{1}{2}Q_{ij}(\bm\theta),
\end{align*}
where $\phi_k=U\bm e_k$.

Hence, we have proved Theorem \ref{1}. 
However, this proof does not take full advantage of the structure of unitary matrices. We will give an alternative proof in Section \ref{sec:mean}, which gives a slightly stronger version of Theorem \ref{1} and will be helpful to establish concentration bounds. 
%\FloatBarrier
\section{Proof of expectation}\label{sec:mean}
\subsection{Notation}

Let us first introduce some notation used throughout the proof.

Let $\Phi$ be the canonical identification from $\CC^N$ to $\RR^{2N}$, that is, $\Phi$ takes a complex vector $\psi=\mathbf{x}+i\mathbf{y}$ and maps it to a real vector $\mathbf{z}=(\mathbf{x}^\top,\mathbf{y}^\top)^\top$ of twice the dimension:
\begin{align*}
  \Phi\colon \mathbb{C}^N &\to \mathbb{R}^{2N} \\
  \psi = (x_1+iy_1, \dots, x_N+iy_N)^\top &\mapsto \Phi(\psi) = (x_1, \dots, x_N, y_1, \dots, y_N)^\top
\end{align*}
With a slight abuse of notation, we define the homomorphism $\Phi: M_N(\mathbb{C}) \to M_{2N}(\mathbb{R})$ that maps a complex matrix $Z = A + iB$ to its real representation $\Phi(Z) = \bigl[\begin{smallmatrix} A & -B \\ B & A \end{smallmatrix}\bigr]$. Its restriction on $\mathrm U(N)$ is an irreducible real representation to $\mathrm O(2N)\cap Sp(2N,\RR)$. Moreover, for every $Z\in M_N(\CC)$ and $\psi\in\CC^N$, we have $\Phi(Z\psi)=\Phi(Z)\Phi(\psi)$, $\Phi(Z^*)=\Phi(Z)^\top$.

Denote $\langle\cdot, \cdot\rangle_{\CC}$ as the standard (complex) inner product in $\CC^N$ and $\langle\cdot, \cdot\rangle$ as the standard inner product in $\RR^N$ or $\RR^{2N}$.
It is easy to check that $\Phi$ preserves the standard real inner product: 
\begin{equation*} 
\operatorname{Re}(\langle\psi_1,\psi_2\rangle_{\CC})=\langle\Phi(\psi_1),\Phi(\psi_2)\rangle \qquad \forall\,  \psi_1,\psi_2\in\CC^N. 
\end{equation*} 
It is also easy to verify that 
$J=\Phi(iI_N)=\bigl[\begin{smallmatrix}
    0&-I_N\\
    I_N&0
\end{smallmatrix}\bigr]  
$
gives the symplectic matrix in $\RR^{2N}$. 

For parameterized quantum states $\psi_{\bm\theta}$, we will denote $\mathbf{z}_{\bm\theta}=\Phi(\psi_{\bm\theta})$ and $\mathbf{x}_{\bm\theta}=\operatorname{Re} (\psi_{\bm\theta}),\mathbf{y}_{\bm\theta}=\operatorname{Im}(\psi_{\bm\theta})$. 
Thus, the Jacobian of $\mathbf{z}_{\bm\theta}$ with respect to $\bm\theta$ is denoted as $\frac{\partial \mathbf{z}_{\bm\theta}}{\partial\bm\theta}\in\RR^{2N\times m}$. 
We also denote $\bm{p}=(p_1,\cdots,p_N)^\top$ where $p_i=\frac{x_i^2+y_i^2}{\|\mathbf{x}\|_2^2+\|\mathbf{y}\|_2^2}$ as the probability distribution corresponding to the quantum state $\psi$ observed in the standard basis. 

Finally, we write  $P(\mathbf{z}_1,\cdots,\mathbf{z}_k)$ as the orthogonal projection onto the subspace $\spanop\{\mathbf{z}_1,\cdots,\mathbf{z}_k\}$.

\subsection{Representation of QFIM and CFIM}

The first step of our proof is to characterize QFIM and CFIM with respect to $\psi_{\bm\theta}$ as a vector in $\mathbb{R}^{2N}$. The following two lemmas offer  alternative representations that would simplify subsequent computations.
\begin{lemma}
\label{lem:QFIM}
    The quantum Fisher information matrix defined in Definition \ref{def:QFIM} is equivalent to the following definition:
    \begin{equation}
        \label{QFIM:def2}
        Q_{ij}(\bm\theta)=\left\langle\frac{\partial \bm{z}_{\bm\theta}}{\partial\theta_i},\frac{\partial \bm{z}_{\bm\theta}}{\partial\theta_j}\right\rangle-\left\langle\frac{\partial \bm{z}_{\bm\theta}}{\partial\theta_i},\bm{z}_{\bm\theta}\right\rangle\left\langle\frac{\partial \bm{z}_{\bm\theta}}{\partial\theta_j},\bm{z}_{\bm\theta}\right\rangle-\left\langle\frac{\partial \bm{z}_{\bm\theta}}{\partial\theta_i},J\bm{z}_{\bm\theta}\right\rangle\left\langle\frac{\partial \bm{z}_{\bm\theta}}{\partial\theta_j},J\bm{z}_{\bm\theta}\right\rangle,
    \end{equation}
    or written in matrix form:
    \begin{equation}
    \label{QFIM:def2:matrix form}
        Q(\bm\theta)=\left(\frac{\partial \bm{z}_{\bm\theta}}{\partial\bm\theta}\right)^\top \bigl(I-P(\bm{z}_{\bm\theta},J\bm{z}_{\bm\theta})\bigr)\frac{\partial \bm{z}_{\bm\theta}}{\partial\bm\theta}.
    \end{equation}
\end{lemma}
\begin{proof}
    Let us first check (\ref{QFIM:def2}). Write $\frac{\partial\psi_{\bm\theta}}{\partial\theta_i}=\frac{\partial\bm{x}_{\bm\theta}}{\partial\theta_i}+i\frac{\partial\bm{y}_{\bm\theta}}{\partial\theta_i}$ and $\psi_{\bm\theta}=\bm{x}_{\bm\theta}+i\bm{y}_{\bm\theta}$, then substitute them back in (\ref{eq:QGT}), we have
    \begin{align*}
        Q_{ij}(\bm\theta)=&\left(\left\langle\frac{\partial\bm{x}_{\bm\theta}}{\partial\theta_i},\frac{\partial\bm{x}_{\bm\theta}}{\partial\theta_j}\right\rangle+\left\langle\frac{\partial\bm{y}_{\bm\theta}}{\partial\theta_i},\frac{\partial\bm{y}_{\bm\theta}}{\partial\theta_j}\right\rangle\right)-\left(\left\langle\frac{\partial\bm{x}_{\bm\theta}}{\partial\theta_i},\bm{x}_{\bm\theta}\right\rangle+\left\langle\frac{\partial\bm{y}_{\bm\theta}}{\partial\theta_i},\bm{y}_{\bm\theta}\right\rangle\right)\left(\left\langle\frac{\partial\bm{x}_{\bm\theta}}{\partial\theta_j},\bm{x}_{\bm\theta}\right\rangle+\left\langle\frac{\partial\bm{y}_{\bm\theta}}{\partial\theta_j},\bm{y}_{\bm\theta}\right\rangle\right)\\
        &-\left(\left\langle\frac{\partial\bm{x}_{\bm\theta}}{\partial\theta_i},\bm{y}_{\bm\theta}\right\rangle-\left\langle\frac{\partial\bm{y}_{\bm\theta}}{\partial\theta_i},\bm{x}_{\bm\theta}\right\rangle\right)\left(\left\langle\frac{\partial\bm{x}_{\bm\theta}}{\partial\theta_j},\bm{y}_{\bm\theta}\right\rangle-\left\langle\frac{\partial\bm{y}_{\bm\theta}}{\partial\theta_j},\bm{x}_{\bm\theta}\right\rangle\right)\\
        &=\left\langle\frac{\partial \bm{z}_{\bm\theta}}{\partial\theta_i},\frac{\partial \bm{z}_{\bm\theta}}{\partial\theta_j}\right\rangle-\left\langle\frac{\partial \bm{z}_{\bm\theta}}{\partial\theta_i},\bm{z}_{\bm\theta}\right\rangle\left\langle\frac{\partial \bm{z}_{\bm\theta}}{\partial\theta_j},\bm{z}_{\bm\theta}\right\rangle-\left\langle\frac{\partial \bm{z}_{\bm\theta}}{\partial\theta_i},J\bm{z}_{\bm\theta}\right\rangle\left\langle\frac{\partial \bm{z}_{\bm\theta}}{\partial\theta_j},J\bm{z}_{\bm\theta}\right\rangle.
    \end{align*}
    Hence (\ref{QFIM:def2}) holds. Note that $\bm{z}_{\bm\theta},J\bm{z}_{\bm\theta}$ are orthogonal in $\RR^{2N}$, so \[P(\bm{z}_{\bm\theta},J\bm{z}_{\bm\theta})=P(\bm{z}_{\bm\theta})+P(J\bm{z}_{\bm\theta})=\bm{z}_{\bm\theta}\bm{z}_{\bm\theta}^\top+J\bm{z}_{\bm\theta}\bm{z}_{\bm\theta}^\top J^\top.\]
    Using the fact that the $i$-th column of $\bm{z}_{\bm\theta}^\top\frac{\partial \bm{z}_{\bm\theta}}{\partial\bm\theta}$ is $\langle\frac{\partial \bm{z}_{\bm\theta}}{\partial\theta_i},\bm{z}_{\bm\theta}\rangle$ and the $i$-th column of $\bm{z}_{\bm\theta}^\top J^\top\frac{\partial \bm{z}_{\bm\theta}}{\partial\bm\theta}$ is $\langle\frac{\partial \bm{z}_{\bm\theta}}{\partial\theta_i},J\bm{z}_{\bm\theta}\rangle$, we can verify (\ref{QFIM:def2:matrix form}) directly by checking each of its components.
\end{proof}
\begin{lemma}
\label{lem:CFIM}
    If  each component of $\bm{p}_{\bm\theta}$ is non-zero, then the classical Fisher information matrix under standard basis defined in (\ref{CFIM:def}) is equivalent to the following definition:
    \begin{equation}
        \label{CFIM:def2}
        F^I_{ij}(\bm\theta)=\left\langle\frac{\partial \bm{z}_{\bm\theta}}{\partial\theta_i},\frac{\partial \bm{z}_{\bm\theta}}{\partial\theta_j}\right\rangle-\left\langle\frac{\partial \bm{z}_{\bm\theta}}{\partial\theta_i},\bm{z}_{\bm\theta}\right\rangle\left\langle\frac{\partial \bm{z}_{\bm\theta}}{\partial\theta_j},\bm{z}_{\bm\theta}\right\rangle-\sum\limits_{k=1}^{N}\frac{\left\langle\frac{\partial \bm{z}_{\bm\theta}}{\partial\theta_i},D_kJ\bm{z}_{\bm\theta}\right\rangle\left\langle\frac{\partial \bm{z}_{\bm\theta}}{\partial\theta_j},D_k J\bm{z}_{\bm\theta}\right\rangle}{\left\langle D_k J\bm{z}_{\bm\theta},D_k J\bm{z}_{\bm\theta}\right\rangle},
    \end{equation}
    where $D_k=\operatorname{diag}(\bm{e}_k+\bm{e}_{k+N})$, $\bm{e}_k$ are the standard basis vectors in $\RR^{2N}$. It also has the matrix form:
    \begin{equation}
    \label{CFIM:def2:matrix form}
        F^I(\bm\theta)=\left(\frac{\partial \bm{z}_{\bm\theta}}{\partial\bm\theta}\right)^\top \bigl(I-P(\bm{z}_{\bm\theta},D_1J\bm{z}_{\bm\theta},\cdots,D_NJ\bm{z}_{\bm\theta})\bigr)\frac{\partial \bm{z}_{\bm\theta}}{\partial\bm\theta}.
    \end{equation}
\end{lemma}
\begin{proof}
    For simplicity, in this proof we denote $\frac{\partial}{\partial\theta_i}$ as $\partial_i$ and omit the index $\bm\theta$ if there is no ambiguity. Expand the right side of (\ref{CFIM:def}) using $p_k=x_k^2+y_k^2$, we know that 
    \begin{align*}
        F^I_{ij}&=\sum\limits_{k=1}^{N}\left(\frac{\partial_i x_k^2+\partial_iy_k^2}{2\sqrt{p_k}}\right)\left(\frac{\partial_j x_k^2+\partial_jy_k^2}{2\sqrt{p_k}}\right)\\
        &=\sum\limits_{k=1}^{N}\left(\frac{x_k^2}{p_k}\partial_ix_k\partial_j x_k+\frac{y_k^2}{p_k}\partial_iy_k\partial_j y_k+\frac{x_ky_k}{p_k}(\partial_i x_k\partial_j y_k+\partial_j x_k\partial_i y_k)\right).
    \end{align*}
    Since $\|\bm{z}\|_2^2=1$, we have $2\langle\partial_i\bm{z},\bm{z}\rangle=\partial_i\|\bm{z}\|_2^2=0$. It is easy to check that $\|D_k J\bm{z}\|_2^2=p_k$. Now we expand the right side of (\ref{CFIM:def2}) as follows:
    \begin{align*}
     RHS&=\sum\limits_{k=1}^{N}\left(\partial_i x_k\partial_j x_k+\partial_i y_k\partial_j y_k\right)-\sum\limits_{k=1}^{N}\frac{(\partial_i y_k\cdot x_k-\partial_i x_k\cdot y_k)(\partial_j y_k\cdot x_k-\partial_j x_k\cdot y_k)}{p_k}\\
     &=\sum\limits_{k=1}^{N}\left(\partial_i x_k\partial_j x_k+\partial_i y_k\partial_j y_k\right)-\sum\limits_{k=1}^{N}\frac{\partial_i y_k\partial_j y_k \cdot x_k^2+\partial_i x_k\partial_j x_k\cdot y_k^2}{p_k}+\sum\limits_{k=1}^{N}\frac{x_ky_k}{p_k}(\partial_i x_k\partial_j y_k+\partial_j x_k\partial_i y_k)\\
     &=\sum\limits_{k=1}^{N}\left(\frac{x_k^2}{p_k}\partial_ix_k\partial_j x_k+\frac{y_k^2}{p_k}\partial_iy_k\partial_j y_k+\frac{x_ky_k}{p_k}(\partial_i x_k\partial_j y_k+\partial_j x_k\partial_i y_k)\right)=F^I_{ij}.
    \end{align*}
    Therefore, (\ref{CFIM:def2}) is equivalent to (\ref{CFIM:def}). Note that $\bm{z},D_1J\bm{z},\cdots,D_NJ\bm{z}$ are orthogonal in $\RR^{2N}$, so \[P(\bm{z},D_1J\bm{z},\cdots,D_NJ\bm{z})=P(\bm{z})+\sum\limits_{k=1}^{N}P(D_k J\bm{z}).\] Then, similar to Lemma~\ref{lem:QFIM}, we can check element-wise that (\ref{CFIM:def2:matrix form}) holds.
\end{proof}
\begin{coro}
    If each component of $U^*\psi_{\bm\theta}$ is non-zero, then the classical Fisher information matrix under basis $U$ defined in (\ref{CFIM:def}) is equivalent to the following definition:
    \begin{align}
    \label{CFIM:basis form}
    F^U(\bm\theta)=\left(\frac{\partial \bm{z}_{\bm\theta}}{\partial\bm\theta}\right)^\top V^\top(I_{2N}-P(V\bm{z}_{\bm\theta},D_1JV\bm{z}_{\bm\theta},\cdots,D_NJV\bm{z}_{\bm\theta}))V\frac{\partial \bm{z}_{\bm\theta}}{\partial\bm\theta},\quad V=\Phi(U)^\top.
\end{align}
\end{coro}
\begin{proof}
    $F^U(\bm\theta)$ is equivalent to the CFIM under the standard basis with respect to the quantum state $U^*\psi_{\bm\theta}$. Hence, we can derive the matrix form of $F^U(\bm\theta)$ by replacing $\Phi(\psi_{\bm\theta})=\bm{z}_{\bm\theta}$ with $\Phi(U^*\psi_{\bm\theta})=V\bm{z}_{\bm\theta}$ in (\ref{CFIM:def2:matrix form}). (\ref{CFIM:basis form}) is then directly obtained by the replacement.
\end{proof}

\subsection{Expectation of projections}

From (\ref{CFIM:basis form}), we know that to compute the expectation of $F^U(\bm\theta)$, it suffices to compute the expectation of each projection $V^\top P(D_kJV\bm{z}_{\bm\theta})V=P(V^\top D_kJV\bm{z}_{\bm\theta})$. For every $\bm z\in\mathbb S^{2N-1}$, denote $\bm P_k^U(\bm z)=P(\Phi(U)D_kJ\Phi(U)^\top\bm z)$. The next lemma derives conditional expectation results for $\bm P_k^U(\bm z)$.
\begin{lemma}
     \label{thm: conditional expectation of Pk}
     Let $U$ be a random unitary matrix that is generated by Haar distribution $\mu_H$ on $\mathrm U(N)$. For every $1\leq k\leq N$, $\psi,\bm r\in \mathbb S_\CC^{N-1}$ such that $\bm r$ has no zero entry, we have 
    \begin{equation}
        \label{eq:EPk}
        \sum\limits_{k=1}^N\EE_{U\sim\mu_H}[\bm P_k^U(\Phi(\psi))|U^*\psi=\bm r]=\frac{1}{2}\left(I_{2N} - P(\Phi(\psi)) + P(J\Phi(\psi))\right).
    \end{equation}
\end{lemma}
\begin{proof}
    We first consider the case when $\psi=\bm e_1$ and $k=1$. Let $U=[\bm{u}_1,\cdots,\bm{u}_N]$, $\tilde D_1=\operatorname{diag}(\bm e_1)\in\CC^{N\times N}$, $\bm{u}_1=\bm{x}_1+i\bm{y}_1$ and $\bm{x}_1=(x_{11},\cdots,x_{N1})^\top$, $\bm{y}_1=(y_{11},\cdots,y_{N1})^\top$, then 
    \[\Phi(U)D_1\Phi(U)^\top=\Phi(U\tilde D_1U^*)=\Phi(\bm{u}_1\bm{u}_1^*)=\begin{bmatrix}
        \bm{x}_1\bm{x}_1^\top+\bm{y}_1\bm{y}_1^\top&\bm{x}_1\bm{y}_1^\top-\bm{y}_1\bm{x}_1^\top\\
        -(\bm{x}_1\bm{y}_1^\top-\bm{y}_1\bm{x}_1^\top)&\bm{x}_1\bm{x}_1^\top+\bm{y}_1\bm{y}_1^\top
    \end{bmatrix}.\]
    Hence, $\bm{w}_1\triangleq \Phi(U)D_1\Phi(U)^\top J\bm e_1=[y_{11}\bm{x}_1^\top-x_{11}\bm{y}_1^\top,x_{11}\bm{x}_1^\top+y_{11}\bm{y}_1^\top]^\top$. Since $\|\bm{x}_1\|_2^2+\|\bm{y}_1\|_2^2=1$, we can know that $\|\bm{w}_1\|_2^2=x_{11}^2+y_{11}^2$. Note that $J$ commutes with $D_k$ and $\Phi(U)$. As a result, we can write $\bm P_1^U(\bm e_1)$ as     
    \[\begin{aligned} 
    & \bm P_1^U(\bm e_1)=\bm w_1\bm w_1^\top/\|\bm w_1\|_2^2 = \begin{bmatrix} A & C \\ D & B \end{bmatrix} \triangleq
    \\
    & \frac{1}{x_{11}^2+y_{11}^2}\begin{bmatrix}  y_{11}^2\bm{x}_1\bm{x}_1^\top+x_{11}^2\bm{y}_1\bm{y}_1^\top-x_{11}y_{11}(\bm{x}_1\bm{y}_1^\top+\bm{y}_1\bm{x}_1^\top)&y_{11}^2\bm{x}_1\bm{y}_1^\top-x_{11}^2\bm{y}_1\bm{x}_1^\top+x_{11}y_{11}(\bm{x}_1\bm{x}_1^\top-\bm{y}_1\bm{y}_1^\top)\\
        y_{11}^2\bm{y}_1\bm{x}_1^\top-x_{11}^2\bm{x}_1\bm{y}_1^\top+x_{11}y_{11}(\bm{x}_1\bm{x}_1^\top-\bm{y}_1\bm{y}_1^\top)&y_{11}^2\bm{y}_1\bm{y}_1^\top+x_{11}^2\bm{x}_1\bm{x}_1^\top+x_{11}y_{11}(\bm{x}_1\bm{y}_1^\top+\bm{y}_1\bm{x}_1^\top)
    \end{bmatrix}.
    \end{aligned}\]
     The distribution of $U$ implies that  $\bm{u}_1$ follows the uniform distribution in $\mathbb S^{2N-1}$. The condition $U^*\bm e_1=\bm r$ implies that the first row of $U$ is fixed. It is known that the conditional distribution of $(x_{21},\cdots,x_{N1},y_{21},\cdots,y_{N1})$ is the uniform distribution on the sphere $(1-x_{11}^2-y_{11}^2)\mathbb S^{2N-3}$. Now we calculate the expectation of $\bm P_1^U(\bm e_1)$ for each element in the matrix. From now on, in this proof, we may write $\EE[X]$ for $\EE[X|U^*\psi=\bm r]$ for brevity. We consider each block separately.

    For the top-left block $A$, we have 
    \[A_{ij}=\frac{y_{11}^2x_{i1}x_{j1}+x_{11}^2y_{i1}y_{j1}-x_{11}y_{11}(x_{i1}y_{j1}+x_{j1}y_{i1})}{x_{11}^2+y_{11}^2},\quad 1\leq i,j\leq N.\]
    It is well-known that for a vector $\bm{v}$ following the uniform distribution in $\mathbb S^d$, we have $\EE[v_I|v_{J}]=0$ for index set $I\cap J=\varnothing$. So when $i\neq j$, one can always find that $\EE[A_{ij}]=0$. When $i=j\neq1$, $\EE[A_{ij}]=\EE\left[\frac{y_{11}^2x_{i1}^2+x_{11}^2y_{i1}^2}{x_{11}^2+y_{11}^2}\right]=\frac{1}{2}\EE[x_{i1}^2+y_{i1}^2]=\frac{1-x_{11}^2-y_{11}^2}{2N-2}$(Exchange $x_{i1},y_{i1}$). When $i=j=1$, $A_{ij}=0$. Hence, $\EE[A]=\frac{1-x_{11}^2-y_{11}^2}{2N-2}(I_N-E_{11})$, where $E_{ij}$ has only one non-zero entry at $(i,j)$ with value $1$.

    For the bottom-right block $B$, we have
    \[B_{ij}=\frac{y_{11}^2y_{i1}y_{j1}+x_{11}^2x_{i1}x_{j1}+x_{11}y_{11}(x_{i1}y_{j1}+x_{j1}y_{i1})}{x_{11}^2+y_{11}^2},\quad 1\leq i,j\leq N.\]
    We still have $\EE[B_{ij}]=0$ for $i\neq j$. If $i=j\neq 1$, then similarly we have $\EE[B_{ij}]=\frac{1-x_{11}^2-y_{11}^2}{2N-2}$. If $i=j=1$, $B_{ij}=x_{11}^2+y_{11}^2$. Hence, $\EE[B]=\frac{1-x_{11}^2-y_{11}^2}{2N-2}(I_N-E_{11})+(x_{11}^2+y_{11}^2)E_{11}$.

    For the top-right block $C$, we have
    \[C_{ij}=\frac{y_{11}^2x_{i1}y_{j1}-x_{11}^2y_{i1}x_{j1}+x_{11}y_{11}(x_{i1}x_{j1}-y_{j1}y_{i1})}{x_{11}^2+y_{11}^2},\quad 1\leq i,j\leq N.\]
    We still have $\EE[C_{ij}]=0$ for $i\neq j$. When $i=j\neq1$, $\EE[C_{ij}]=\EE\left[\frac{x_{11}y_{11}(x_{i1}^2-y_{i1}^2)}{x_{11}^2+y_{11}^2}\right]=0$(Exchange $x_{i1},y_{i1}$). When $i=j=1$, $C_{ij}=0$. Hence, $\EE[C]=0$. 
    
    For the bottom-left block $D$, $\EE[D]=0$ since $D=C^\top$.
    
    Hence, 
    \[\EE[\bm P_1^U(\bm e_1)]=\frac{1}{2(N-1)}(1-x_{11}^2-y_{11}^2)(I_{2N}-E_{11}-E_{N+1,N+1})+(x_{11}^2+y_{11}^2)E_{N+1,N+1}.\] 
    Similarly, for all $1\leq k\leq N$, we have 
    \[\EE[\bm P_k^U(\bm e_1)]=\frac{1}{2(N-1)}(1-x_{1k}^2-y_{1k}^2)(I_{2N}-E_{11}-E_{N+1,N+1})+(x_{1k}^2+y_{1k}^2)E_{N+1,N+1}.\]
    Sum up $\bm P_k^U(\bm e_1)$ together and notice that $\sum\limits_{k=1}^{N}(x_{1k}^2+y_{1k}^2)=1$, we have
    \begin{align*}
        \sum\limits_{k=1}^N\EE[\bm P_k^U(\Phi(\bm e_1))]&=\sum\limits_{k=1}^{N}\frac{1-x_{1k}^2-y_{1k}^2}{2(N-1)}(I_{2N}-E_{11}-E_{N+1,N+1})+\sum\limits_{k=1}^N(x_{1k}^2+y_{1k}^2)E_{N+1,N+1}\\
        &=\frac{1}{2}\left(I_{2N}-E_{11}+E_{N+1,N+1}\right)\\
        &=\frac{1}{2}\left(I_{2N}-P(\bm e_1)+P(J\bm e_1)\right).
    \end{align*}
    Now we extend our result to arbitrary $\psi$. Note that $\bm P_k^U(\Phi(\psi))=P(\Phi(iU\tilde D_k U^*\psi))$. Pick and fix a $U_0\in \mathrm U(N)$ such that $U_0^*\bm e_1=\psi$ and write $U=U_0^*U'$ where $U'=U_0U$. Then \[\bm P_k^U(\Phi(\psi))=P(\Phi(iU\tilde D_kU^*U_0^*\bm e_1))=P(\Phi(iU_0^*U'\tilde D_kU'^*\bm e_1))=\Phi(U_0^*)P(\Phi(iU'\tilde D_kU'^*\bm e_1))\Phi(U_0).\]
    The condition $U^*\psi=\bm r$ is equivalent to $U'^*\bm e_1=\bm r$. By the left invariance of Haar distribution, we have $\EE_{U\sim\mu_H}[\bm P_k^U(\Phi(\psi))|U^*\psi =\bm r]=\EE_{U'\sim\mu_H}[\Phi(U_0^*)\bm P_k^{U'}(\Phi(\bm e_1)) \Phi(U_0)|U'^*\bm e_1=\bm r]$, so 
    \[ \sum\limits_{k=1}^N\EE_{U\sim\mu_H}[\bm P_k^U(\Phi(\psi))]=\frac{1}{2}\Phi(U_0^*)(I_{2N}-P(\bm e_1)+P(J\bm e_1)) \Phi(U_0)=\frac{1}{2}\left(I_{2N}-P(\Phi(\psi))+P(J \Phi(\psi))\right).
    \qedhere \]
\end{proof}
Using Lemma \ref{lem:QFIM}, \ref{thm: conditional expectation of Pk} and (\ref{CFIM:basis form}), we can now directly prove \thmref{1}.
\begin{proof}[Proof of \thmref{1}]
    By Lemma \ref{lem:QFIM} and (\ref{CFIM:basis form}), we have 
    \[F^U(\bm\theta)=Q(\bm\theta)-\left(\frac{\partial \bm{z}_{\bm\theta}}{\partial\bm\theta}\right)^\top\left(\sum\limits_{k=1}^N\bm P_k^U(\bm z_{\bm\theta})-P(J\bm z_{\bm\theta})\right)\frac{\partial \bm{z}_{\bm\theta}}{\partial\bm\theta}.\]
    Since $U^*\psi_{\bm\theta}$ has no zero entry almost surely, Lemma \ref{thm: conditional expectation of Pk} implies that 
    \[\sum\limits_{k=1}^N\bm \EE_{U\sim\mu_H}[P_k^U(\bm z_{\bm\theta})]=\EE_{\bm r}\left[\EE_{U\sim\mu_H}\left[\sum\limits_{k=1}^N\bm P_k^U(\bm z_{\bm\theta})|U^*\psi_{\bm\theta}=\bm r\right]\right]=\frac{1}{2}\left(I_{2N}-P(\bm z_{\bm\theta})+P(J\bm z_{\bm\theta})\right),\] 
    so we have
    \begin{align*}
        \EE_{U\sim\mu_H}[F^U(\bm\theta)]&=Q(\bm\theta)-\left(\frac{\partial \bm{z}_{\bm\theta}}{\partial\bm\theta}\right)^\top\left(\sum\limits_{k=1}^N\bm \EE_{U\sim\mu_H}[P_k^U(\bm z_{\bm\theta})]-P(J\bm z_{\bm\theta})\right)\frac{\partial \bm{z}_{\bm\theta}}{\partial\bm\theta}\\
        &=Q(\bm\theta)-\left(\frac{\partial \bm{z}_{\bm\theta}}{\partial\bm\theta}\right)^\top\left(\frac{1}{2}\left(I_{2N}-P(\bm z_{\bm\theta})+P(J\bm z_{\bm\theta})\right)-P(J\bm z_{\bm\theta})\right)\frac{\partial \bm{z}_{\bm\theta}}{\partial\bm\theta}\\
        &=\frac{1}{2}Q(\bm\theta)=\frac{1}{2}\operatorname{Re}(\mathcal{Q}(\bm\theta)). \qedhere
    \end{align*}
\end{proof}
\begin{remark}
    From the proof of \thmref{1} and Lemma \ref{thm: conditional expectation of Pk}, we know that $\EE_{U\sim\mu_H}[F^U(\bm\theta)|U^*\psi_{\bm\theta}=\bm r]=\frac{1}{2}Q(\bm\theta)$, which is stronger than the claim of the theorem. 
\end{remark}

\section{Proof of variance and moments}\label{sec:variance}
In this section, we derive the variance and estimate moments of $F^U(\bm\theta)$, thus proving \thmref{2},\ref{thm:6}. We start by rewriting $Q(\bm\theta)$ and $F^U(\bm\theta)$ in a more compact form based on the results of Lemma \ref{lem:QFIM}, \ref{lem:CFIM}.
\begin{notation}
     For any $U\in \mathrm U(N)$ and $\bm z\in\RR^{2N}$, we define two subspaces:
    \[V(\bm{z}) \triangleq \spanop \{\bm{z},J\bm{z}\},\] 
    \[S^U(\bm{z}) \triangleq \spanop \{\bm z,\Phi(U)D_1\Phi(U)^\top J\bm{z},\cdots,\Phi(U) D_N\Phi(U)^\top J\bm{z}\},\]
    and denote $P_{V}(\bm z),P_{S^U}(\bm z)$ be the orthogonal projections onto the subspaces $V(\bm{z}),S^U(\bm{z})$ respectively. Denote for $\bm u,\bm v\in\RR^{2N}$, 
    \begin{equation*}
    X(\bm u,\bm v) \triangleq \EE_{U\sim\mu_H}[(\frac{1}{2}I_{2N}-P_{S^U}(\bm u))\bm v\bm v^\top(\frac{1}{2}I_{2N}-P_{S^U}(\bm u))].
    \end{equation*}
\end{notation}
\begin{prop}
    Denote $A(\bm\theta)=(I_{2N}-P_V(\bm z_{\bm\theta}))\frac{\partial\bm z_{\bm\theta}}{\partial\bm\theta}$, then we have 
    \begin{equation}
        \label{eq:QFIM&CFIM}
        Q(\bm\theta)=A(\bm\theta)^\top A(\bm\theta),\quad F^U(\bm\theta)=A(\bm\theta)^\top (I_{2N}-P_{S^U}(\bm z_{\bm\theta}))A(\bm\theta).
    \end{equation}
\end{prop}
\begin{proof}
    For any orthogonal projection $P_V$ where $V$ is the projection subspace, we have $P^2_V=P_V$. For two subspaces $V\subset W$, we have $P_VP_W=P_WP_V=P_{V}$. Therefore, the first equality in (\ref{eq:QFIM&CFIM}) holds since $((I_{2N}-P_V(\bm z_{\bm\theta}))^2=(I_{2N}-P_V(\bm z_{\bm\theta}))$. Note that $S^U(\bm z)\supset V(\bm z)$, so $(I_{2N}-P_V(\bm z_{\bm\theta}))(I_{2N}-P_{S^U}(\bm z_{\bm\theta}))(I_{2N}-P_V(\bm z_{\bm\theta}))=I_{2N}-P_{S^U}(\bm z_{\bm\theta})$. Then it is easy to check the second equality in (\ref{eq:QFIM&CFIM}) by (\ref{CFIM:basis form}).
\end{proof}
Denote $\bm v_i$ as the $i$-th column of $A(\bm\theta)$. Based on (\ref{eq:QFIM&CFIM}) and \thmref{1}, we know that
\begin{equation}
    \label{eq:Variance}
    \var_{U\sim\mu_H}[F_{ij}^U(\bm\theta)]= \EE \left(\bm v_i^\top (\frac{1}{2}I_{2N}-P_{S^U}(\bm z_{\bm\theta}))\bm v_j\right)^2=\bm v_i^\top X(\bm z_{\bm\theta},\bm v_j)\bm v_i.
\end{equation}
Thus, the problem is converted to calculate the matrix $X(\bm z_{\bm\theta},\bm v_j)$. The next lemma suggests that to calculate $X(\bm u,\bm v)$ such that $\langle\Phi^{-1}(\bm u),\Phi^{-1}(\bm v)\rangle_{\CC}=0$, we only need to know $X(\bm e_1,\bm e_2)$.
\begin{lemma}
\label{lem:X(u,v) under rotation}
    For any vector $\bm{u},\bm{v}\in\RR^{2N}$ and $U_0\in \mathrm U(N)$, we have 
    \begin{equation}
        X(\bm{u},\bm{v})=\Phi(U_0)^\top X(\Phi(U_0)\bm{u},\Phi(U_0)\bm{v})\Phi(U_0).
    \end{equation}
\end{lemma}
\begin{proof}
    Denote $\bm r=\Phi(U_0)\bm{u},\bm s=\Phi(U_0)\bm{v}$. Since $P_{S^U}(\bm{u})=\Phi(U_0)^\top P_{S^{UU_0^*}}(\Phi(U_0)\bm{u})\Phi(U_0)$, we have
    \begin{align*}
        X(\bm{u},\bm{v})&=\EE_{U\sim\mu_H}\left[\left(\frac{1}{2}I_{2N}-\Phi(U_0)^\top P_{S^{UU_0^*}}(\bm{r})\Phi(U_0)\right)\bm{v}\bm{v}^\top\left(\frac{1}{2}I_{2N}-\Phi(U_0)^\top P_{S^{UU_0^*}}(\bm{r})\Phi(U_0)\right)\right]\\
        &=\Phi(U_0)^\top\EE_{UU_0^*\sim\mu_H}\left[\left(\frac{1}{2}I_{2N}-P_{S^{UU_0^*}}(\bm r)\right)\bm s\bm s^\top\left(\frac{1}{2}I_{2N}-P_{S^{UU_0^*}}(\bm r)\right)\right]\Phi(U_0)\\
        &=\Phi(U_0)^\top X(\bm{r},\bm s)\Phi(U_0)\\
        &=\Phi(U_0)^\top X(\Phi(U_0)\bm{u},\Phi(U_0)\bm{v})\Phi(U_0),
    \end{align*}
where the last equality follows from the left-invariance of Haar distribution.
\end{proof}
Suppose that $N\geq 2$, if $\langle\Phi^{-1}(\bm u),\Phi^{-1}(\bm v)\rangle_{\CC}=0$, then we can find $U_0\in \mathrm U(N)$ such that $\Phi(U_0)\bm e_1=\bm u,\Phi(U_0)\bm e_2=\bm v$. By Lemma \ref{lem:X(u,v) under rotation}, we can derive $X(\bm u,\bm v)$ from $X(\bm e_1,\bm e_2)$. Note that 
\[\bm z_{\bm\theta}^\top A(\bm\theta)=\bm z_{\bm\theta}^\top (I_{2N}-P_V(\bm z_{\bm\theta}))\frac{\partial\bm z_{\bm\theta}}{\partial\bm\theta}=0,\quad (J\bm z_{\bm\theta})^\top A(\bm\theta)=(J\bm z_{\bm\theta})^\top (I_{2N}-P_V(\bm z_{\bm\theta}))\frac{\partial\bm z_{\bm\theta}}{\partial\bm\theta}=0,\] so $\bm v_i^\top \bm z_{\bm\theta}=\bm v_i^\top J\bm z_{\bm\theta}=0$ and thus $\langle\Phi^{-1}(\bm z_{\bm\theta}),\Phi^{-1}(\bm v_i)\rangle_{\CC}=0$. The only left thing is to derive $X(\bm e_1,\bm e_2)$.
\begin{lemma}
\label{lem:X(e1,e2)}
    Suppose that $N\geq 2$, then $X(\bm e_1,\bm e_2)=\frac{1}{8N}(I_{2N}+P(\bm e_2,J\bm e_2)-P(\bm e_1,J\bm e_1))$.
\end{lemma}
\begin{proof}
    For simplicity, in this proof we denote $X=X(\bm e_1,\bm e_2)\in\RR^{2N\times2N}$.  Note that $P_{S^U}(\bm e_1)\bm e_1=\bm e_1$, $P_{S^U}(\bm e_1)J\bm e_1=J\bm e_1$, so for every $\bm{w}\in \spanop\{\bm e_1,J\bm e_1\}$, we have $\bm{w}^\top(\frac{1}{2}I_{2N}-P_{S^U}(\bm e_1))\bm e_2=0$. This implies $X_{ij}=0$ if one of $i,j$ belongs to $\{1,N+1\}$.
    
    Let $T_1=\{U\in \mathrm U(N):U\bm e_1=\bm e_1,U\bm e_2=\bm e_2\}$ be a subgroup of $\mathrm U(N)$.
    For every unitary $U\in T_1$, we have $X\Phi(U)=\Phi(U)X$. Note that $T_1$ is an irreducible representation of the linear isomorphisms in the subspace $W=\spanop\{\bm e_1,J\bm e_1,\bm e_2,J\bm e_2\}^\perp$. Hence, by Schur Lemma, we know that there is a constant $c$ such that $X\bm{w}=c\bm{w}$ for all $\bm{w}\in W$. This implies that $X_{ij}=0$ for $i,j\notin\{1,2,N+1,N+2\}$ and $i\neq j$, and $X_{ii}=c$ for $i\notin\{1,2,N+1,N+2\}$.

    Now the only possible non-zero off-diagonal elements of $X$ is $X_{2,N+2}=X_{N+2,2}$. Denote $R=\operatorname{diag}(I_N,-I_N)$, then $\Phi(\bar{U})=R\Phi(U)R$, since $U,\bar U$ follows the same distribution, and 
    \[P(\Phi(\bar U)^\top D_k\Phi(\bar U)J\bm e_1)=P(R\Phi(U)^\top RD_kR\Phi(U)RJ\bm e_1)=RP(\Phi(U)^\top D_k\Phi(U)J\bm e_1)R,\]
    so we have $\frac{1}{2}I_{2N}-P_{S^{\bar U}}(\bm e_1)=R(\frac{1}{2}I_{2N}-P_{S^U}(\bm e_1))R$, and then 
    \[X(\bm e_1,\bm e_2)=\EE_{\bar U\sim\mu_H}[R(\frac{1}{2}I_{2N}-P_{S^U}(\bm e_1))R\bm e_2\bm e_2^\top R(\frac{1}{2}I_{2N}-P_{S^U}(\bm e_1))]=RX(\bm e_1,\bm e_2)R.\]
    Hence, $\bm e_2^\top X(\bm e_1,\bm e_2)J\bm e_2=\bm e_2^\top RX(\bm e_1,\bm e_2)RJ\bm e_2=-\bm e_2^\top X(\bm e_1,\bm e_2)J\bm e_2$ by $R\bm e_2=\bm e_2, RJ\bm e_2=-J\bm e_2$. This implies $X_{2,N+2}=X_{N+2,2}=0$. Similarly, we have $X_{22}=X_{N+2,N+2}$.

    Based on the above argument, we know that $X$ must be in the form of \[\operatorname{diag}(0,a,c,\cdots,c,0,a,c\cdots,c),\]
    where each $c,\cdots,c$ is of length $N-2$. Note that $\operatorname{trace}((\frac{1}{2}I_{2N}-P_{S^U}(\bm{u}))\bm{v}\bm{v}^\top(\frac{1}{2}I_{2N}-P_{S^U}(\bm{u})))=\frac{1}{4}\|\bm{v}\|^2$ since $(\frac{1}{2}I_{2N}-P_{S^U}(\bm{u}))^2=\frac{1}{4}I$, so $2a+(2N-4)c=\frac{1}{4}$. If $N=2$, then the proof is finished. So let us assume that $N\geq 3$. In this case, we claim that $a=2c$.  

    Denote $C_U=\frac{1}{2}I_{2N}-P_{S^U}(\bm e_1)$ for brevity. The form of $X$ implies that $\EE_{U\sim\mu_H}[(\bm e_2^\top C_U\bm e_2)^2]=a$, $\EE_{U\sim\mu_H}[(\bm e_3^\top C_U\bm e_2)^2]=c$ and $\EE_{U\sim\mu_H}[(\bm e_2^\top C_U\bm e_2)(\bm e_2^\top C_U\bm e_3)]=0$. By the rotation invariance of Haar distribution, for any unit vectors $\bm{u},\bm{v}$ such that $\bm u,\bm v\in \spanop\{\bm e_1,J\bm e_1\}^\perp$ and $\bm{u}\perp\bm{v},\bm{u}\perp J\bm{v}$, we have 
    \[\EE_{U\sim\mu_H}[(\bm{u}^\top C_U\bm{u})^2]=a, \quad \EE_{U\sim\mu_H}[(\bm{u}^\top C_U\bm{v})^2]=c,\quad\EE_{U\sim\mu_H}[(\bm{v}^\top C_U\bm{u})(\bm{u}^\top C_U\bm{u})]=0. \]
    Consider $\EE_{U\sim\mu_H}[(\bm{w}^\top C_U\bm{w})^2]=a$ for $\bm{w}=\frac{1}{\sqrt2}(\bm e_2+\bm e_3)$. Expand the expectation in the term related to $\bm e_2,\bm e_3$, we can find that $a-2c=\EE_{U\sim\mu_H}[(\bm e_2^\top C_U\bm e_2)(\bm e_3^\top C_U\bm e_3)]$. Denote $U=(U_{ij})_{N\times N}$ and $U_{ij}=r_{ij}e^{i\theta_{ij}}$. 
    It can be directly calculated from the matrix form of $\bm P_k^U(\bm e_1)$  (replacing $x_{ij},y_{ij}$ by $r_{ij}\cos\theta_{ij},r_{ij}\sin\theta_{ij}$ respectively for $A_{22},A_{33}$ in Lemma \ref{thm: conditional expectation of Pk} ) that
    \begin{equation}
    \label{eq:Q_u}
       \bm e_2^\top C_U\bm e_2=-\frac{1}{2}\sum\limits_{k=1}^{N}r_{2k}^2\cos(2(\theta_{2k}-\theta_{1k})),\quad\bm e_3^\top C_U\bm e_3=-\frac{1}{2}\sum\limits_{l=1}^{N}r_{3l}^2\cos(2(\theta_{3l}-\theta_{1l})). 
    \end{equation}
    Denote $D_{\phi_1,\phi_2}=\operatorname{diag}(1,e^{i\phi_1},e^{i\phi_2},1,\cdots,1)\in \mathrm U(N)$, where $\phi_1,\phi_2$ are independent random phases uniformly distributed in $[0,2\pi]$. Then $\EE_{U\sim\mu_H}[f(U)]=\EE_{U\sim\mu_H}\EE_{\phi_1,\phi_2}[f(D_{\phi_1,\phi_2}U)]$ for every function $f$ by the left invariance of Haar distribution. For each pair $(k,l)$, we have 
    \[\int_0^{2\pi}\int_0^{2\pi}\cos(2(\theta_{2k}-\theta_{1k}+\phi_1))\cos(2(\theta_{3l}-\theta_{1l}+\phi_2))d\phi_1d\phi_2=0.\]
    As a result, $\EE_{U\sim\mu_H}[r_{2k}^2\cos(2(\theta_{2k}-\theta_{1k}))r_{3l}^2\cos(2(\theta_{3l}-\theta_{1l}))]$
    and $\EE_{U\sim\mu_H}[(\bm e_2^\top C_U\bm e_2)(\bm e_3^\top C_U\bm e_3)]$ are 0. This implies $a=2c$. From the two linear equations $a=2c,2a+(2N-4)c=\frac{1}{4}$, we know that $a=\frac{1}{4N},c=\frac{1}{8N}$. Then the proof is complete. 
\end{proof}
\begin{coro}
    Suppose that $N\geq 2$ and $\bm{u},\bm{v}$ are unit vectors such that $\langle\bm{u},\bm{v}\rangle=0,\langle\bm{u},J\bm{v}\rangle=0$, then 
    \begin{equation}
        \label{eq:X(u,v) final}
        X(\bm{u},\bm{v})=\frac{1}{8N}(I_{2N}+P(\bm{v},J\bm{v})-P(\bm{u},J\bm{u})).
    \end{equation}
\end{coro}
\begin{proof}
    Since $\langle\bm{u},\bm{v}\rangle=0,\langle\bm{u},J\bm{v}\rangle=0$, there is a unitary $U_0$ such that $U_0\bm e_1=\Phi^{-1}(\bm{u}),U_0\bm e_2=\Phi^{-1}(\bm{v})$. Hence, by Lemma \ref{lem:X(u,v) under rotation} and \ref{lem:X(e1,e2)}, we have 
    \[X(\bm{u},\bm{v})=\Phi(U_0) X(\bm e_1,\bm e_2)\Phi(U_0)^\top=\frac{1}{8N}(I_{2N}+P(\bm{v},J\bm{v})-P(\bm{u},J\bm{u})).
    \qedhere \]
\end{proof}
Having figured out the expression of $X(\bm u,\bm v)$, we are able to derive the variance for each entry of the CFIM. Denote $\tilde Q(\bm\theta)=A(\bm\theta)^\top JA(\bm\theta)$. Then we have $\operatorname{Im}(\mathcal{Q}(\bm\theta))=\tilde Q(\bm\theta)$. This is because $\Phi(\bm{u})^\top J\Phi(\bm{v})=\operatorname{Im}(\langle\bm{u},\bm{v}\rangle_{\CC})$. It can be verified that $\operatorname{Im}(\mathcal{Q}(\bm\theta))=(\frac{\partial\bm{z}_{\bm\theta}}{\partial\bm\theta})^\top J(I-P(\bm{z}_{\bm\theta},J\bm{z}_{\bm\theta}))\frac{\partial\bm{z}_{\bm\theta}}{\partial\bm\theta}$. $J$ and $P(\bm{z}_{\bm\theta},J\bm{z}_{\bm\theta})$ commute, so we have $\operatorname{Im}(\mathcal{Q}(\bm\theta))=A(\bm\theta)^\top JA(\bm\theta)$. Note that $\mathcal{Q}(\bm\theta)$ is Hermitian, so \[\mathcal{Q}_{ii}(\bm\theta)=Q_{ii}(\bm\theta),\quad \mathcal{Q}_{ij}(\bm\theta)\mathcal{Q}_{ji}(\bm\theta)=Q_{ij}^2(\bm\theta)+\tilde Q_{ij}^2(\bm\theta).\] Therefore, to prove \thmref{2}, it suffices to prove that 
\[\var_{U\sim\mu_H}[F_{ij}^U(\bm\theta)]=\frac{1}{8N}(Q_{ii}(\bm\theta)Q_{jj}(\bm\theta)+{Q}_{ij}(\bm\theta)^2+\tilde {Q}_{ij}(\bm\theta)^2).\]
\begin{proof}[Proof of \thmref{2}]
From (\ref{eq:QFIM&CFIM}), we know that $Q_{ij}(\bm\theta)=\bm v_i^\top\bm v_j$ and similarly $\tilde Q_{ij}(\bm\theta)=\bm v_i^\top J\bm v_j$.
We have shown that $\bm v_j^\top\bm{z}_{\bm{\theta}}=\bm v_j^\top J\bm{z}_{\bm{\theta}}=0$, so by (\ref{eq:Variance}) and (\ref{eq:X(u,v) final}), we have 
    \begin{align*}
        \var[F_{ij}^U(\bm{\theta})]&=\|\bm v_j\|_2^2\bm v_i^\top X(\bm{z}_{\bm{\theta}},\frac{\bm v_j}{\|\bm v_j\|_2})\bm v_i\\
        &=\frac{\|\bm v_j\|_2^2}{8N}\left(\bm v_i^\top(I-P(\bm{z}_{\bm{\theta}},J\bm{z}_{\bm{\theta}}))\bm v_i+\frac{(\bm v_i^\top\bm v_j)^2}{\|\bm v_j\|_2^2}+\frac{(\bm v_i^\top J\bm v_j)^2}{\|\bm v_j\|_2^2}\right)\\
        &=\frac{Q_{jj}(\bm{\theta})}{8N}\left(Q_{ii}(\bm{\theta})+\frac{Q_{ij}(\bm{\theta})^2}{Q_{jj}(\bm{\theta})}+\frac{\tilde Q_{ij}^2(\bm{\theta})}{Q_{jj}(\bm{\theta})}\right)\\
        &=\frac{1}{8N}(Q_{ii}(\bm{\theta})Q_{jj}(\bm{\theta})+Q_{ij}^2(\bm{\theta})+\tilde Q_{ij}^2(\bm{\theta})).
    \end{align*}
    From this element-wise result, it is easy to check that the matrix form result in \thmref{2} holds.
\end{proof}

Now let us prove \thmref{thm:6}. We first consider the case when $i=j$. In this case, the moments can be accurately calculated. Note that $X_{ii}=\bm v_i^\top (\frac{1}{2}I_{2N}-P_{S^U}(\bm z_{\bm\theta}))\bm v_i/\|\bm v_i\|_2^2$ and $\langle \bm z_{\bm\theta},\bm v_i\rangle=\langle\bm z_{\bm\theta},J\bm v_i\rangle=0$, it suffices to take $\bm v_i=\bm e_2,\bm z_{\bm\theta}=\bm e_1$ by rotation invariance. 

\begin{lemma}
\label{lem:Q-ineq}
For any $a^2+b^2=1,a,b\in\mathbb R$, denote
\[X(U)=a\bm e_2^\top (\frac{1}{2}I_{2N}-P_{S^U}(\bm e_1))\bm e_2+b\bm e_2^\top (\frac{1}{2}I_{2N}-P_{S^U}(\bm e_1))J\bm e_2,\] then we have
    \begin{equation}
          \EE[X^{2k+1}(U)]=0, \quad \EE[X^{2k}(U)]= \frac{(2k)!}{k!8^k}\cdot\left[\prod_{j=0}^{k-1}(N+2j)\right]^{-1}. 
    \end{equation}

\end{lemma}
\begin{proof}
    
Denote $C_U=\frac{1}{2}I_{2N}-P_{S^U}(\bm e_1)$.
 Conditioned on $U^*\bm e_{1}=\bm r$, for $V=\mathrm U(N)$ such that $V^*\bm e_1=\bm e_1$, we have 
 \[X(VU)=\bm e_2^\top \Phi(V) (aC_U+bC_UJ)\Phi(V)^\top \bm e_2=\bm e_2^\top \Phi(V)\left (aC_U+\frac{b}{2}(C_UJ-JC_U)\right)\Phi(V)^\top \bm e_2.\]
 
 Fix $U=U_0$, $\Phi(V)^\top \bm e_2$ follows the uniform distribution on the unit sphere in $W=\operatorname{span}\{\bm e_1, J\bm e_1\}^\perp$.
  Then by our argument, $X(U)|U^*\bm e_1=\bm r$ has the same distribution as $\bm u^\top A \bm u$ where $\bm u$ follows the uniform distribution on the unit sphere in $W$, $A=[aC_U+\frac{b}{2}(C_UJ-JC_U)]|_W$  as the restriction of the linear operator to the subspace $W$. 
  
  It is easy to check that \footnote{First show that $C_U^2=(C_UJ)^2=(JC_U)^2=\frac{1}{4}I_{2N}-P(\bm e_1,J\bm e_1)$.}
 \begin{equation}
 \label{eq:JCU}
  \operatorname{tr}(C_U|_W)=0,\quad C_U^2|_W=(JC_U)^2|_W=(C_U J)^2|_W=\frac{1}{4}\operatorname{id}|_W.
 \end{equation}
 From \eqref{eq:JCU}, it can be derived that $\operatorname{tr}(A)=0,A^2=\frac{1}{4}\operatorname{id}|_W$. Hence, 
 \begin{equation}
 \label{eq:Z}
\EE[X(U)^k|U^*\bm e_1=\bm r]=\EE[(\bm u^\top A \bm u)^k]=\frac{\Gamma(N-1)}{2^k\Gamma(N+k-1)} \mathbb{E}[(Z^\top AZ)^k],\quad Z\sim \mathcal N(0,\bm I_{2N-2}).
 \end{equation}

Denote the moment generating function of $Z^\top AZ$ as $M(t)$, then 
\[M(t)=\frac{1}{\sqrt{\operatorname{det}(I_{2N-2}-2tA)}},\quad \ln M(t) = -\frac{1}{2} \ln \det(I_{2N-2} - 2tA).\]
By $\operatorname{tr}(A^{2k+1})=0, \operatorname{tr}(A^{2k})=(2N-2)(\frac{1}{4})^k$ and Taylor expansion, we have
\begin{align*}
\ln M(t) &= -\frac{1}{2} \mathrm{Tr}\left( \ln(I_{2N-2} - 2tA) \right)\\
&= \frac{1}{2} \sum_{j=1}^{\infty} \frac{(2t)^j \mathrm{Tr}(A^j)}{j}\\
&=\frac{1}{2}\sum\limits_{j=1}^{\infty}\frac{(4t^2)^j(\frac{1}{4})^j(2N-2)}{2j}\\
&=-\frac{(N-1)}{2}\ln(1-t^2). 
\end{align*}

So $M(t)=(1-t^2)^{-\frac{N-1}{2}}$, $\mathbb{E}[(Z^\top AZ)^{2k+1}]=0,\mathbb{E}[(Z^\top AZ)^{2k}]=\frac{(2k)!\prod_{j=0}^{k-1} \left(\frac{N-1}{2} + j \right)}{k!}$, as a result of \eqref{eq:Z},
\begin{equation}
\label{eq:XU}
    \EE[X(U)^{2k-1}]=0,\quad \EE[X(U)^{2k}]=
    \frac{\Gamma(N-1)(2k)!\prod_{j=0}^{k-1} \left( \frac{N-1}{2} + j \right)}{\Gamma(N+2k-1)k!2^{2k}}=
    \frac{(2k)!}{k!8^k}\cdot\left[\prod_{j=0}^{k-1}(N+2j)\right]^{-1}.
\end{equation}
\end{proof}
By Stirling formula, there exists $c_1,c_2>0$ such that $c_1\sqrt k\leq (\frac{(2k)!}{k!8^k})^{1/2k}\leq c_2\sqrt k$, and  $\left[\prod_{j=0}^{k-1}(N+2j)\right]^{\frac{1}{2k}}\sim\Theta (\sqrt{N})$ as $N\geq k$. Hence, we prove the theorem for $i=j$. 

For $i\neq j$, we first consider the case when $\bm v_j\perp\bm v_i,\bm v_j\perp J\bm v_i$. By rotation invariance, it suffices to estimate $\EE[\bm (e_2^\top(\frac{1}{2}I_{2N}-P_{S^U}(\bm e_1))\bm e_3)^k]$. We estimate this quantity by the following 2 lemmas.

\begin{lemma}
\label{lem:Q,B}

    Denote $Q_{\bm v}=\bm v^\top(\frac{1}{2}I_{2N}-P_{S^U}(\bm e_1))\bm v, B_{\bm u,\bm v}=\bm u^\top(\frac{1}{2}I_{2N}-P_{S^U}(\bm e_1))\bm v $, $\tilde Q_{\bm u}=\bm u^\top(\frac{1}{2}I_{2N}-P_{S^U}(\bm e_1))J\bm u$. Then for all unit vectors  $\bm u,\bm v\in\operatorname{span}\{\bm e_1,J\bm e_1\}^\perp$, such that $\bm u\perp \bm v,\bm u\perp J\bm v$, and $i,j\in \mathbb N$, such that either of them is odd, we have
    \begin{equation}
       \EE[Q_{\bm u}^iQ_{\bm v}^j]=\EE[B^i_{\bm u,\bm v}(aQ_{\bm u}+b\tilde Q_{\bm u})^j]=0.\quad \forall a,b\in\mathbb R 
    \end{equation}
\end{lemma}
\begin{proof}
    By rotational invariance, we can assume $\bm u=\bm e_2,\bm v=\bm e_3$. As \eqref{eq:Q_u} has shown, we have
    \[   Q_{\bm e_2}=-\frac{1}{2}\sum\limits_{k=1}^{N}r_{2k}^2\cos(2(\theta_{2k}-\theta_{1k})),\quad Q_{\bm e_3}=-\frac{1}{2}\sum\limits_{l=1}^{N}r_{3l}^2\cos(2(\theta_{3l}-\theta_{1l})). \]
    Similarly, 
    \[B_{\bm e_2,\bm e_3}=-\sum\limits_{m=1}^{N}r_{2m}r_{3m}\left(\frac{1}{2}\cos(\theta_{2m}-\theta_{3m})+\frac{1}{2}\cos(\theta_{2m}+\theta_{3m}-2\theta_{1m})\right),\]
\[\tilde  Q_{\bm e_2}=-\frac{1}{2}\sum\limits_{n=1}^{N}r_{2n}^2\sin(2(\theta_{2n}-\theta_{1n})).\]
    Assume $i$ is odd, write $Q_{\bm e_2}(\phi)=\sum\limits_{k=1}^{N}r_{2k}^2\cos(2(\theta_{2k}-\theta_{1k}+\phi))$ where $\phi$ is independent random phases uniformly distributed in $[0,2\pi]$. As in the proof of variance, we have 
\[\EE_{U\sim\mu_H}[Q_{\bm e_2}^iQ_{\bm e_3}^j]=\EE_{U\sim\mu_H}\EE_{\phi}[Q_{\bm e_2}(\phi)^iQ_{\bm e_3}^j].\]
Each term after expanding the inner expectation is of the form $\EE[\sin^a \phi \cos ^b\phi] \cdot f(U) $, it must be 0 since $a+b$ is odd. Hence, $\EE_{U\sim\mu_H}[Q_{\bm e_2}^iQ_{\bm e_3}^j]=0$. If $j$ is odd, do the same for $Q_{\bm e_3}(\phi)=\sum\limits_{l=1}^{N}r_{3l}^2\cos(2(\theta_{3l}-\theta_{1l}+\phi))$.

For $\EE[B^i_{\bm u,\bm v}(aQ_{\bm u}+b\tilde Q_{\bm u})^j]$, if $i$ is odd, do the same for $Q_{\bm e_3}(\phi)$; otherwise $i$ is even and $j$ is odd, do the same for $aQ_{\bm e_2}(\phi)+b\tilde Q_{\bm e_2}(\phi)$. We can prove the objective is 0 by the same argument.
\end{proof}

\begin{lemma}
    For any unit vectors $\bm u,\bm v\in\operatorname{span}\{\bm e_1,J\bm e_1\}^\perp$,$\bm u\perp \bm v,\bm u\perp J\bm v$,  we have 
    \begin{equation}
       \EE[B_{\bm u,\bm v}^{2k+1}]=0, \quad \frac{1}{2^{2k-1}}\EE[Q_{\bm u}^{2k}]\leq \EE[B_{\bm u,\bm v}^{2k}]\leq \frac{1}{2}\EE[Q_{\bm u}^{2k}]. 
    \end{equation}
    
\end{lemma}
\begin{proof}
    By rotational invariance, $\EE[B_{\bm u,\bm v}^{k}]=\EE[B_{\bm u',\bm v'}^{k}]$ where $\bm u'=\frac{1}{\sqrt 2}(\bm u+\bm v),\bm v'=\frac{1}{\sqrt{2}}(\bm u-\bm v)$. Since $B_{\bm u',\bm v'}=\frac{1}{2}(Q_{\bm u}-Q_{\bm v})$, we have $\EE[B_{\bm u,\bm v}^{k}]=\frac{1}{2^k}\EE[(Q_{\bm u}-Q_{\bm v})^k]$. Applying Lemma \ref{lem:Q,B}, we have $\EE[B_{\bm u,\bm v}^{2k+1}]=0$ and
    \[\EE[B_{\bm u,\bm v}^{2k}]=\frac{1}{2^{2k}}\sum\limits_{j=0}^{k}\binom{2k}{2j}\EE\left[Q_{\bm u}^{2j}Q_{\bm v}^{2(k-j)}\right]\geq\frac{1}{2^{2k-1}}\EE[Q_{\bm u}^{2k}],\]
    \begin{align*}
        \EE[B_{\bm u,\bm v}^{2k}]&=\frac{1}{2^{2k}}\sum\limits_{j=0}^{k}\binom{2k}{2j}\EE\left[Q_{\bm u}^{2j}Q_{\bm v}^{2(k-j)}\right]\\
        &\leq \frac{1}{2^{2k}}\sum\limits_{j=0}^{k}\binom{2k}{2j}\left(\frac{j}{k}\EE[Q_{\bm u}^{2k}]+\frac{k-j}{k}\EE[Q_{\bm v}^{2k}]\right)\\
        &=\frac{1}{2}\EE[Q_{\bm u}^{2k}],
    \end{align*}
    where the last equality follows from the fact that $\EE[Q_{\bm u}^{2k}]=\EE[Q_{\bm v}^{2k}]=\EE[X(U)^{2k}]$.
\end{proof}

Finally, let us deal with general $\bm u,\bm v\in\operatorname{span}\{\bm e_1,J\bm e_1\}^\perp$.

\begin{lemma}
\label{lem:B-ineq}
    For any unit vectors $\bm u,\bm v\in\operatorname{span}\{\bm e_1,J\bm e_1\}^\perp$,  we have 
    \begin{equation}
       \EE[B_{\bm u,\bm v}^{2k+1}]=0, \quad \frac{1}{5^{k}}\EE[Q_{\bm u}^{2k}]\leq \EE[B_{\bm u,\bm v}^{2k}]\leq {2^k}\EE[Q_{\bm u}^{2k}]. 
    \end{equation}
    
\end{lemma}
\begin{proof}
   Write $\bm v=a\bm u+bJ\bm u+c\bm w$ where $\bm w\perp \bm u,\bm w\perp J\bm u$ and $a^2+b^2+c^2=1$. Denote $X(U)=\frac{a}{\sqrt{a^2+b^2}}Q_{\bm u}+\frac{b}{\sqrt{a^2+b^2}}\tilde Q_{\bm u}$.  Then  
   \begin{align*}
       B_{\bm u,\bm v}^k&=(aQ_{\bm u}+b\tilde Q_{\bm u}+cB_{\bm u,\bm w})^k\\
       &=\sum\limits_{j=0}^k\binom{k}{j}\left(\frac{a}{\sqrt{a^2+b^2}}Q_{\bm u}+\frac{b}{\sqrt{a^2+b^2}}\tilde Q_{\bm u}\right)^jB_{\bm u,\bm w}^{k-j} \left(\sqrt{a^2+b^2}\right)^jc^{k-j}\\
       &=\sum\limits_{j=0}^k\binom{k}{j}X(U)^jB_{\bm u,\bm w}^{k-j}\left(\sqrt{1-c^2}\right)^jc^{k-j}.
   \end{align*}
   Applying Lemma \ref{lem:Q,B}, we have $\EE[B_{\bm u,\bm v}^{2k+1}]=0$ and
   \begin{align*}
       \EE[B_{\bm u,\bm v}^{2k}]&=\sum\limits_{j=0}^k\binom{2k}{2j}\EE\left[X(U)^{2j}B_{\bm u,\bm w}^{2(k-j)}\right]\left({1-c^2}\right)^j(c^2)^{k-j}\\
       &\leq \sum\limits_{j=0}^k\binom{2k}{2j}\left({1-c^2}\right)^j(c^2)^{k-j}\EE[Q_{\bm u}^{2k}]\\
       &\leq \left(c+\sqrt{1-c^2}\right)^{2k}\EE[Q_{\bm u}^{2k}]\\
       &\leq 2^k\EE[Q_{\bm u}^{2k}].
   \end{align*}
   On the other hand, we have
\begin{align*}
    \EE[B_{\bm u,\bm v}^{2k}]&=\sum\limits_{j=0}^k\binom{2k}{2j}\EE\left[X(U)^{2j}B_{\bm u,\bm w}^{2(k-j)}\right]\left({1-c^2}\right)^j(c^2)^{k-j}\\
    &\geq c^{2k}\EE\left[B_{\bm u,\bm w}^{2k}\right]+(1-c^2)^k\EE\left[X(U)^{2k}\right]\\
    &\geq \left[\left(\frac{c}{2}\right)^{2k}+(1-c^2)^k\right]\EE[Q_{\bm u}^{2k}]\\
    &\geq \frac{1}{5^k}\EE[Q_{\bm u}^{2k}]. \qedhere
\end{align*} 
\end{proof}
 Combining Lemma \ref{lem:Q-ineq} and Lemma \ref{lem:B-ineq}, we finish the proof of \eqref{eq:momentsk}.  For \eqref{eq:upper-expmom}, we have
 \begin{align*}
     \EE[\exp(\lambda X_{ij})]&=\sum\limits_{k=0}^{\infty}\frac{\lambda^k}{k!}\EE[X_{ij}^k]\\
     &\leq\sum\limits_{k=0}^{\infty}2^k\cdot\frac{\lambda^{2k}}{(2k)!}\cdot \frac{(2k)!}{k! 8^k}\cdot \left[\prod_{j=0}^{k-1}(N+2j)\right]^{-1}\\
     &\leq \sum\limits_{k=0}^{\infty}\left(\frac{\lambda^2}{4N}\right)^k\cdot \frac{1}{k!}\\
     &=\exp\left(\frac{\lambda^2}{4N}\right).
 \end{align*}
It remains to prove \eqref{eq:lower-expmom}. Note that for $t<k+1$, we have
\begin{align*}
    e^t&=\sum\limits_{i=0}^{k}\frac{t^i}{i!}+\sum\limits_{i\geq k+1}\frac{t^i}{i!}\\
    &\leq \sum\limits_{i=0}^{k}\frac{t^i}{i!} +\frac{t^{k+1}}{(k+1)!}\sum\limits_{i=0}^{\infty}\left(\frac{t}{k+1}\right)^i\\
    &= \sum\limits_{i=0}^{k}\frac{t^i}{i!} +\frac{t^{k+1}}{(k+1)!}\cdot\frac{1}{1-\frac{t}{k+1}}.
\end{align*}
 Hence, for $|\lambda|\leq N$, we have
 \begin{align*}
     \EE[\exp(\lambda X_{ij})]&\geq\sum\limits_{k=0}^{\infty}5^{-k}\cdot\frac{\lambda^{2k}}{(2k)!}\cdot \frac{(2k)!}{k! 8^k}\cdot \left[\prod_{j=0}^{k-1}(N+2j)\right]^{-1}\\
     &\geq \sum\limits_{k=0}^{N-1}\left(\frac{\lambda^2}{120N}\right)^k\cdot \frac{1}{k!}+\left(\frac{\lambda^2}{120N}\right)^N\cdot\frac{1}{N!}\cdot\frac{(3N)^N}{\prod_{j=0}^{N-1}(N+2j)}\\
     &\geq \sum\limits_{k=0}^{N-1}\left(\frac{\lambda^2}{120N}\right)^k\cdot \frac{1}{k!}+\left(\frac{\lambda^2}{120N}\right)^N\cdot\frac{1}{N!}\cdot\frac{1}{1-\frac{\lambda^2}{120N^2}} \\
     &\geq \exp\left(\frac{\lambda^2}{120N}\right).
 \end{align*}
 Then the proof of Theorem \ref{thm:6} is complete.
\section{Proof of concentration bounds}\label{sec:concentration}
Concentration bounds on compact Lie groups have been extensively studied. A comprehensive overview of this topic is provided in \cite{meckes2019random}. One of the most powerful tools is the log-Sobolev inequality (LSI), which holds on compact manifolds with positive Ricci curvature. Utilizing the LSI, we can derive concentration bounds for any Lipschitz continuous function on the unitary group. We will use the following lemma in our proofs. 
\begin{lemma}[\cite{meckes2019random}]
\label{lem:concentration}
    Let $f:\mathrm U(d)\to \RR$ be a function such that $|f(U)-f(V)|\leq L\|U-V\|_F$ for all $U,V\in \mathrm U(d)$. Then for every $t>0$, we have
    \begin{equation}
        \label{eq:concentration1}
        \PP(|f(U)-\EE[f(U)]|\geq t)\leq 2\exp\left(-\frac{dt^2}{12L^2}\right).
    \end{equation}
\end{lemma}
One may attempt to show that $F^U(\bm\theta)$ is Lipschitz continuous. However, this is generally not true. The reason is that $\sqrt{\bm p^U(\bm\theta)}$ may not be differentiable at $\bm\theta$ where the vector $\bm p^U(\bm\theta)$ has zero entries, even if $\bm p^U(\bm\theta)$ is differentiable with respect to $\bm\theta$. As a result, $F^U(\bm\theta)$ exhibits discontinuity at certain parameter values. Nevertheless, as in the following lemma, we claim that when $U$ satisfies the condition the same as in Lemma \ref{thm: conditional expectation of Pk}, $F_{ij}^U(\bm\theta)$ is continuous with Lipschitz constant independent of the dimension $N$. 
\begin{lemma}
\label{lem:keyLip}
    Given a parameter $\bm\theta$ and a unit vector $\bm r$ that has no zero entry, for any $U,V\in \mathrm U(N)$ such that $U^*\psi_{\bm\theta}=V^*\psi_{\bm\theta}=\bm r$, we have 
    \[\|P_{S^U}(\bm z_{\bm\theta})-P_{S^V}(\bm z_{\bm\theta})\|_2\leq \frac{\pi}{2}\|U-V\|_F,\quad \|P_{S^U}(\bm z_{\bm\theta})-P_{S^V}(\bm z_{\bm\theta})\|_F\leq \frac{\sqrt 2\pi}{2}\|U-V\|_F.\]
\end{lemma}
\begin{proof}
    Denote $\bm r=(r_1,\cdots,r_N)^\top$ then $\|\Phi(U) D_k\Phi(U)^\top J\bm{z}_{\bm\theta}\|_2^2=|r_k|^2\neq 0$. Let $\bm{w}_k=\Phi(iU\tilde D_k\bm r)$,then
\begin{equation}
    P(\Phi(U) D_k\Phi(U)^\top J\bm{z}_{\bm\theta})=P(\bm{w}_k)=|r_k|^{-2}\bm{w}_k\bm{w}_k^\top.
\end{equation}
 Let $\delta U,\delta\bm{w}_k$ be the small variation of $U,\bm{w}_k$ respectively. By their definitions, it is easy to check that $\delta\bm{w}_k=\Phi(ir_k\delta U\bm e_k)$. Moreover, the tangent space of $\mathrm U(N)$ at $U$ is $\{UY:Y^*=-Y\}$. Hence, $\delta U=U\delta Y+O(\delta^2)$ where $\delta Y=\delta\cdot Y$ is a skew-Hermitian matrix.  $\delta\bm{w}_k=\Phi(ir_kU\delta Y\bm e_k)+O(\delta^2)$, so we have 
\begin{equation}
\label{eq:Project error}
    P(\bm{w}_k+\delta\bm{w}_k)-P(\bm{w}_k)=\frac{1}{|r_k|^2}(\bm{w}_k\delta\bm{w}_k^\top+\delta\bm{w}_k\bm{w}_k^\top)+O(\delta^2).
\end{equation}

Denote $\tilde{\bm{w}}_k=\frac{\bm{w}_k}{|r_k|},\delta\tilde{\bm{w}}_k=\frac{\delta\bm{w}_k}{|r_k|}$. Then $\|\tilde{\bm{w}}_k\|_2=1,\|\delta\tilde{\bm{w}}_k\|_2=\delta\|Y\bm e_k\|_2+O(\delta^2),\delta\tilde{\bm{w}}_k^\top\tilde{\bm{w}}_k=O(\delta^2)$, 
\begin{equation}
    \frac{1}{|r_k|^2}(\bm{w}_k\delta\bm{w}_k^\top+\delta\bm{w}_k\bm{w}_k^\top)=(\tilde{\bm{w}}_k\delta\tilde{\bm{w}}_k^\top+\delta\tilde{\bm{w}}_k\tilde{\bm{w}}_k^\top)+O(\delta^2).
\end{equation}

Let $W=[\tilde{\bm{w}}_1,\cdots,\tilde{\bm{w}}_N, J\tilde{\bm{w}}_1,\cdots,J\tilde{\bm{w}}_N]$, $\delta W=[\delta\tilde{\bm{w}}_1,\cdots,\delta\tilde{\bm{w}}_N,0,\cdots,0]^\top$. Then 
\begin{equation}
\label{eq:error matrix form}
   \sum\limits_{k=1}^N(\tilde{\bm{w}}_k\delta\tilde{\bm{w}}_k^\top+\delta\tilde{\bm{w}}_k\tilde{\bm{w}}_k^\top)=W\delta W+(\delta W)^\top W^\top.
\end{equation}

 Since $\delta\bm{w}_k^\top\bm{w}_l=\Phi(ir_kU\delta Y\bm e_k)^\top\Phi(ir_lU\bm e_l)+O(\delta^2)$, and $\langle ir_kU\delta Y\bm e_k,ir_l U\bm e_l\rangle_{\CC}=r_k\bar r_l\delta Y_{lk}$, we have
 \begin{equation}
     \delta\tilde{\bm{w}}_k^\top\tilde{\bm{w}}_l=\delta\operatorname{Re}(\frac{r_k\bar r_l}{|r_kr_l|}Y_{lk})+O(\delta^2).
 \end{equation}

Similarly, we have 
\begin{equation}
    \delta\tilde{\bm{w}}_k^\top J\tilde{\bm{w}}_l=\delta\operatorname{Im}(\frac{r_k\bar r_l}{|r_kr_l|}Y_{lk})+O(\delta^2).
\end{equation}

Let $\tilde Y_{kl}=\frac{r_k\bar r_l}{|r_kr_l|}Y_{lk}$. Note that $\tilde Y$ is also skew-Hermitian, so $\operatorname{Re}(\tilde Y)$ is anti-symmetric and $\operatorname{Im}(\tilde Y)$ is symmetric. Using the above results for  $\delta\tilde{\bm{w}}_k^\top \tilde{\bm{w}}_l,\delta\tilde{\bm{w}}_k^\top J\tilde{\bm{w}}_l$, we can calculate each entry of the matrix $\delta W W+W^\top(\delta W)^\top$. We get
\begin{equation}
    \delta W W+W^\top(\delta W)^\top=\delta \begin{bmatrix}
    0&\operatorname{Im}(\tilde Y)\\
   \operatorname{Im}(\tilde Y)&0
\end{bmatrix}+O(\delta^2),
\end{equation}
Note that due to the anti-symmetry of $\operatorname{Re}(\tilde Y)$, the top-left block is zero.
\begin{align}
\label{ineq:norm bound}
   \|\delta WW+W^\top(\delta W)^\top\|_2&\leq\delta\|\operatorname{Im}(\tilde Y)\|_2+O(\delta^2)\leq \delta\|\tilde Y\|_2+O(\delta^2).
\end{align}

Now we can bound $\|P_{S^{U+\delta U}}-P_{S^U}(\bm{z}_{\bm\theta})\|_2$ as follows:
\begin{align}
\label{eq:final key norm bound}
    \|P_{S^{U+\delta U}}(\bm{z}_{\bm\theta})-P_{S^U}(\bm{z}_{\bm\theta})\|_2&=\|\sum\limits_{k=1}^{N}(P(\bm{w}_k+\delta\bm{w}_k)-P(\bm{w}_k))\|_2\notag\\
    &=\| W\delta W+(\delta W)^\top W^\top\|_2+O(\delta^2)\notag \\
    &=\|\delta W W+W^\top(\delta W)^\top\|_2+O(\delta^2)\notag\\
    &\leq \delta\|Y\|_2+O(\delta^2).
\end{align}
The first equality follows from (\ref{eq:Project error})$\sim$(\ref{eq:error matrix form}). The second equality follows from the fact that $W$ is an orthogonal matrix in $\RR^{2N\times 2N}$. The last inequality follows from (\ref{ineq:norm bound}) and $\|Y\|_2=\|\tilde Y\|_2$.

Note that $\|\delta U\|_2=\|U\delta Y\|_2+O(\delta^2)=\delta\|Y\|_2+O(\delta^2)$, so we have $\|P_{S^{U+\delta U}}(\bm{z}_{\bm\theta})-P_{S^U}(\bm{z}_{\bm\theta})\|_2\leq\|\delta U\|_2+O(\delta^2)$. Fix $U_0,U_1\in \mathrm U(N)$ such that $U_0^*\psi_{\bm\theta}=\bm e_1,U_1^*\bm r=\bm e_1$, then for every $U$ such that  $U^*\psi_{\bm\theta}=\bm r$, there exists a unique $X\in \mathrm U(N-1)$ such that $U=U_0\begin{bmatrix}
    1&0\\
    0& X
\end{bmatrix}U_1^*$, and vice versa. Denote this map as $U=g(X)$, then it is easy to see that $\|g(X_1)-g(X_2)\|_2=\|X_1-X_2\|_2$ and $\|g(X_1)-g(X_2)\|_F=\|X_1-X_2\|_F$. Let $U=g(X)$ and $V=g(Z)$.
For every $X,Z$ on the compact manifold $\mathrm U(N-1)$, the geodesic distance $d(X,Z)$ between $X,Z$ is no more than $\frac{\pi}{2}\|X-Z\|_F$. Take $M+1$ matrices on the geodesic between $X,Z$, namely $X_0,X_1,\cdots,X_{M}$ such that $X_0=X,X_{M}=Z$ and $\|X_{i+1}-X_i\|_F=O(M^{-1})$, then we have  
\[\|P_{S^U}(\bm z_{\bm\theta})-P_{S^V}(\bm z_{\bm\theta})\|_2\leq\sum\limits_{k=0}^{M-1}\|P_{S^{g(X_k)}}(\bm z_{\bm\theta})-P_{S^{g(X_{k+1})}}(\bm z_{\bm\theta})\|_2\leq \sum\limits_{k=0}^{M-1}\|X_{k+1}-X_{k}\|_F+M\cdot O(M^{-2}).\]
Let $M\to\infty$, we have $\|P_{S^U}(\bm z_{\bm\theta})-P_{S^V}(\bm z_{\bm\theta})\|_2\leq d(X,Z)\leq \frac{\pi}{2} \|X-Z\|_F=\frac{\pi}{2}\|U-V\|_F$.

In a similar way to (\ref{eq:final key norm bound}), we have $\|P_{S^U}(\bm z_{\bm\theta})-P_{S^{U+\delta U}}(\bm z_{\bm\theta})\|_F\leq \sqrt{2}\delta \|Y\|_F+O(\delta^2)$. By the same geodesic argument, we know that $\|P_{S^U}(\bm z_{\bm\theta})-P_{S^V}(\bm z_{\bm\theta})\|_F\leq \sqrt{2}d(X,Z)\leq\frac{\sqrt{2}\pi}{2}\|U-V\|_F$. 
\end{proof}
\begin{remark}
    By the invariance of the Haar distribution of $U=g(X)$, it can be checked that the distribution of $X$ also shares the same invariance. Therefore, 
the conditional distribution on $U^*\psi_{\bm\theta}=\bm r$ is actually a Haar distribution on $\mathrm U(N-1)$.
\end{remark}
Now we can prove \thmref{3}, \ref{4} using Lemma \ref{lem:concentration}, \ref{lem:keyLip} as follows:
\begin{proof}[Proof of \thmref{3}]
    Consider the conditional probability and expectation on $U^*\psi_{\bm\theta}=\bm r$ where $\bm r$ has no zero entry.  
    Denote $\bm v_i,\bm v_j$ as the $i,j$-th column of $A(\bm\theta)$ respectively. Then $Q_{ij}(\bm\theta)=\bm v_i^\top\bm v_j$,  
\begin{align}
\label{eq:Lip evaluation}
    |F^U_{ij}(\bm\theta)-F^V_{ij}(\bm\theta)|/\|\EE_{U\sim\mu_H}[F^U(\bm\theta)]\|_{\max}&=2|\bm v_i^\top(P_{S^V}(\bm{z}_{\bm\theta})-P_{S^U}(\bm{z}_{\bm\theta}))\bm v_j|/\|Q(\bm\theta)\|_{\max}\notag\\
    &\leq 2\|P_{S^V}(\bm{z}_{\bm\theta})-P_{S^U}(\bm{z}_{\bm\theta})\|_2\|\bm v_i\|_2\|\bm v_j\|_2/\|Q(\bm\theta)\|_{\max}\notag\\
    &\leq 2\|P_{S^V}(\bm{z}_{\bm\theta})-P_{S^U}(\bm{z}_{\bm\theta})\|_2,
\end{align}
where the first equality derives from \thmref{1}, the last inequality follows from the fact that $\|Q(\bm\theta)\|_{\max}=\max\limits_{1\leq i\leq m}\|\bm v_i\|_2^2$. Combining (\ref{eq:Lip evaluation}) and Lemma \ref{lem:keyLip}, 
\[f(g(X))=f(U)\triangleq (F_{ij}^U(\bm\theta)-\frac{1}{2}Q_{ij}(\bm\theta))/\|\EE_{U\sim\mu_H}[F^U(\bm\theta)]\|_{\max}\] is Lipschitz continuous with respect to both $U,X$ with constant $\pi$. 

By Lemma \ref{thm: conditional expectation of Pk}, we have $\EE_{U\sim\mu_H}[f(U)|U^*\psi_{\bm\theta}=\bm r]=0$. 
Applying Lemma \ref{lem:concentration} to function $f(g(X))$, where $X$ follows the Haar distribution on $\mathrm U(N-1)$, we have 
\begin{equation}
\label{eq:i,j}
  \PP\left(|F^U_{ij}(\bm\theta)-\frac{1}{2}Q_{ij}(\bm\theta)|\geq \frac{t}{2}\|Q(\bm\theta)\|_{\max}|U^*\psi_{\bm\theta}=\bm r\right)\leq 2\exp\left(-\frac{(N-1)t^2}{12\pi^2}\right).  
\end{equation}
Since $\PP(\max\limits_{1\leq i\leq l}|A_i|\geq t)\leq \sum\limits_{1\leq i\leq l}\PP(|A_i|\geq t)$ and $\|F^U_{ij}(\bm\theta)-\frac{1}{2}Q(\bm\theta)\|_{\max}=\max\limits_{1\leq i,j\leq m}|F^U_{ij}(\bm\theta)-\frac{1}{2}Q_{ij}(\bm\theta)|$, summing up the probability for all $i,j$ in (\ref{eq:i,j}), we have 
\begin{equation}
    \PP\left(\|F^U(\bm\theta)-\frac{1}{2}Q(\bm\theta)\|_{\max}\geq\frac{t}{2} \|Q(\bm\theta)\|_{\max}|U^*\psi_{\bm\theta}=\bm r\right)\leq 2m^2\exp\left(-\frac{(N-1)t^2}{12\pi^2}\right).
\end{equation}
Then the proof is complete by noticing that $\pi^2\leq 10$.
\end{proof}
\begin{remark}
We can also use results in \thmref{thm:6} to derive the bound.
By the constant derived in the proof of \thmref{thm:6}, tail bound $2m^2\exp\left(-\frac{(N-1)t^2}{12\pi^2}\right)$ can be improved to $2m^2\exp\left(-\frac{Nt^2}{4}\right)$.    
\end{remark}

\begin{proof}[Proof of \thmref{4}]
    Still consider the conditional probability and expectation on $U^*\psi_{\bm\theta}=\bm r$ where $\bm r$ has no zero entry.  Let $f(U)=(F^U_{ij}(\bm\theta)-\frac{1}{2}Q_{ij}(\bm\theta))/\sqrt{Q_{ii}(\bm\theta)Q_{jj}(\bm\theta)}$. Then
    \[|f(U)-f(V)|\leq\|P_{S^V}(\bm{z}_{\bm\theta})-P_{S^U}(\bm{z}_{\bm\theta})\|_2\leq \frac{\pi}{2}\|U-V\|_F\]
    by a similar argument as (\ref{eq:Lip evaluation}).  Applying Lemma \ref{lem:concentration} to $f$, we have
    \begin{equation}
\label{eq2:i,j}
  \PP\left(|f(U)|\geq t|U^*\psi_{\bm\theta}=\bm r\right)\leq 2\exp\left(-\frac{(N-1)t^2}{3\pi^2}\right).  
\end{equation}
From the proof of \thmref{thm:6}, the second moment of $f(U)$ can be bounded as
\[\EE_{U\sim\mu_H}[f(U)^2|U^*\psi_{\bm\theta}=\bm r]\leq\frac{1}{2N}.\]
That is $\EE_{U\sim\mu_H}[|F^U_{ij}(\bm\theta)-\frac{1}{2}Q_{ij}(\bm\theta)|^2|U^*\psi_{\bm\theta}=\bm r]\leq \frac{1}{2N}Q_{ii}(\bm\theta)Q_{jj}(\bm\theta)$. As a consequence, by $\EE[|X|]\leq\EE[X^2]^{1/2}$, we have 
\begin{align}
\label{E[FF]}
    \EE_{U\sim\mu_H}[\|F^U(\bm\theta)-\frac{1}{2}Q(\bm\theta)\|_F|U^*\psi_{\bm\theta}=\bm r]&\leq \sqrt{\sum\limits_{i,j}\EE_{U\sim\mu_H}[|F^U_{ij}(\bm\theta)-\frac{1}{2}Q_{ij}(\bm\theta)|^2|U^*\psi_{\bm\theta}=\bm r]}\notag\\
    &=\sqrt{\frac{1}{2N}\sum\limits_{i,j}Q_{ii}(\bm\theta)Q_{jj}(\bm\theta)}\notag\\
    &=\frac{1}{\sqrt{2N}}\operatorname{tr}(Q(\bm\theta))\notag\\
    &\leq \sqrt{\frac{m}{2N}}\|Q(\bm\theta)\|_F.
\end{align}
The last inequality follows from the fact that $Q(\bm\theta)\succeq 0$. Now let $h(U)=\|F^U(\bm\theta)-\frac{1}{2}Q(\bm\theta)\|_F/\|\EE_{U\sim\mu_H}[F^U(\bm\theta)]\|_F$. Since $\|A^\top BA\|_F\leq \|A^\top A\|_F\|B\|_2$ for symmetric $B$, $F^U(\bm\theta)-F^V(\bm\theta)=A(\bm\theta)^\top(P_{S^U}(\bm{z}_{\bm\theta})-P_{S^V}(\bm{z}_{\bm\theta}))A(\bm\theta)$ and $Q(\bm\theta)=A(\bm\theta)^\top A(\bm\theta)$,
we can bound $|h(U)-h(V)|$ using Lemma \ref{lem:keyLip} as 
\[|h(U)-h(V)|\leq 2\|F^U(\bm\theta)-F^V(\bm\theta)\|_F/\|Q(\bm\theta)\|_F\leq 2\|P_{S^V}(\bm{z}_{\bm\theta})-P_{S^U}(\bm{z}_{\bm\theta})\|_2\leq \pi\|U-V\|_F.\]
(\ref{E[FF]}) indicates that $\EE_{U\sim\mu_H}[h(U)|U^*\psi_{\bm\theta}=\bm r]\leq \sqrt{\frac{2m}{N}}$. Applying Lemma \ref{lem:concentration} to $h$ we have 
    \begin{equation}
\label{Prob F}
  \PP\left(\frac{\|F^U(\bm\theta)-\frac{1}{2}Q(\bm\theta)\|_F}{\|\EE_{U\sim\mu_H}[F^U(\bm\theta)]\|_F}-\sqrt{\frac{2m}{N}}\geq t|U^*\psi_{\bm\theta}=\bm r\right)\leq \exp\left(-\frac{(N-1)t^2}{12\pi^2}\right).  
\end{equation}
Then the proof is complete by Law of Total Probability and $\pi^2\leq 10$.
\end{proof}
\begin{remark}
\label{remark:thm4}
  In fact, $|h(U)-h(V)|\leq 2\|F^U(\bm\theta)-F^V(\bm\theta)\|_F/\|Q(\bm\theta)\|_2$, so the inequality \eqref{Prob F} can be improved by replacing $\|\EE_{U\sim\mu_H}[F^U(\bm\theta)]\|_F$ by $\|\EE_{U\sim\mu_H}[F^U(\bm\theta)]\|_2$. Hence, tail bound in \thmref{4} can be improved from $\exp(-cNt^2)$ to $\exp(-cNr^2t^2)$ where $r=\frac{\|Q(\bm\theta)\|_F}{\|Q(\bm\theta)\|_2}\geq 1$. 
 To evaluate $E[\|F^U(\bm\theta)\|_F]$, we use $E[\|F^U(\bm\theta)\|^2_F]^{1/2}$ as an upper-bound. By \thmref{thm:6}, $E[\|F^U(\bm\theta)\|^2_F]^{1/2}=\Theta(1/\sqrt N)\operatorname{tr(Q(\bm\theta))}$. \hl{Since its entries are all sub-Gaussian, second moment could correctly match the order of first moment, rather than being artificially inflated by rare, extreme outliers. We also bound $\operatorname{tr}(Q(\bm \theta))$ by $\sqrt{m}\|Q(\bm\theta)\|_F$, the factor $\sqrt m$ cannot be improved without additional structural assumptions.} 
\end{remark}

Finally, let us prove \thmref{5}. Our proof borrows from the proof of Dvoretzky's theorem in \cite{meckes2019random}, which uses the covering number strategy (Dudley's entropy) to bound the spectral norm.
\begin{proof}[Proof of \thmref{5}]
     In the proof, we always condition on $U^*\psi_{\bm\theta}=\bm r$ where $\bm r$ has no zero entry. For simplicity, in this proof, the notations $\PP,\EE$ mean the conditional probability and expectation.  
     
     Let $E=\{A(\bm\theta)\bm{y}:\bm{y}\in\RR^{m}\}$ be the image space of $A(\bm\theta)$. Denote 
    \begin{equation}
        \label{aux}
        \mathcal{R}(\bm{z},X)=\bm{z}^\top(\frac{1}{2}I_{2N}-P_{S^U}(\bm{z}_{\bm\theta}))\bm{z},\quad U=g(X),\quad \bm{z}\in E\cap\mathbb{S}^{2N-1}.
    \end{equation}
    Under the imposed condition $U^*\psi_{\bm\theta}=\bm r$, $X$ follows the Haar distribution on $\mathrm U(N-1)$. Denote $\mathcal{R}(X)=\sup\limits_{\bm{z}\in E\cap\mathbb{S}^{2N-1}}|\mathcal{R}(\bm{z},X)|$, then 
    \begin{equation}
        \label{eq:upperbounf:R(X)}
         \mathcal{R}(X)=\sup\limits_{A(\bm\theta)\bm{y}\neq \bm 0}\frac{|\bm{y}^\top(F^U(\bm\theta)-\frac{1}{2}Q(\bm\theta))\bm{y}|}{|\bm{y}^\top Q(\bm\theta)\bm{y}|}\geq \frac{\|(F^U(\bm\theta)-\frac{1}{2}Q(\bm\theta))\|_2}{\|Q(\bm\theta)\|_2}.
    \end{equation}
    From the proof of Lemma \ref{lem:keyLip}, we know that $|\mathcal{R}(\bm{z},X)-\mathcal{R}(\bm{z},Y)|\leq\frac{\pi}{2}\|X-Y\|_F$ for every $X,Y\in \mathrm U(N-1)$. This implies $|\mathcal{R}(X)-\mathcal{R}(Y)|\leq\frac{\pi}{2}\|X-Y\|_F$. By Lemma \ref{lem:concentration}, we have
    \begin{equation}
        \label{eq:concentration for R(X)}
        \PP(\mathcal{R}(X)-\EE[\mathcal{R}(X)]\geq t)\leq \exp\left(-\frac{(N-1)t^2}{30}\right).
    \end{equation}
 For any $\bm{z}_1,\bm{z}_2\in E\cap \mathbb{S}^{2N-1}$, denote $D(\bm{z}_1,\bm{z}_2,X)=(\bm{z}_1+\bm{z}_2)^\top P_{S^U}(\bm{z}_{\bm\theta})(\bm{z}_1-\bm{z}_2)$. Then we have 
    \[\mathcal{R}(\bm{z}_1,X)-\mathcal{R}(\bm{z}_2,X)=\bm{z}_1^\top P_{S^U}(\bm{z}_{\bm\theta})\bm{z}_1-\bm{z}_2^\top P_{S^U}(\bm{z}_{\bm\theta})\bm{z}_2=(\bm{z}_1+\bm{z}_2)^\top P_{S^U}(\bm{z}_{\bm\theta})(\bm{z}_1-\bm{z}_2),\]
    so $D(\bm{z}_1,\bm{z}_2,X)=\mathcal{R}(\bm{z}_1,X)-\mathcal{R}(\bm{z}_2,X)$ and then 
    \[|D(\bm{z}_1,\bm{z}_2,X)-D(\bm{z}_1,\bm{z}_2,Y)|\leq2\|P_{S^U}(\bm{z}_{\bm\theta})-P_{S^V}(\bm{z}_{\bm\theta})\|_2\|\bm{z}_1-\bm{z}_2\|_2\leq \pi\|X-Y\|_F\|\bm{z}_1-\bm{z}_2\|_2.\]
    Note that $\EE[\mathcal{R}(\bm{z},X)]=0$, so $\EE[D(\bm{z}_1,\bm{z}_2,X)]=0$. As a result of Lemma \ref{lem:concentration},
    \begin{equation}
    \label{sub-gau-incre}
        \PP(|\mathcal{R}(\bm{z}_1,X)-\mathcal{R}(\bm{z}_2,X)|\geq t)\leq 2\exp\left (-\frac{(N-1)t^2}{120\|\bm{z}_1-\bm{z}_2\|_2^2}\right).
    \end{equation}
    For any $X\in \mathrm U(N-1)$, view $\mathcal{R}(\bm{z},X)$ as a mean-zero stochastic process over $\bm{z}\in E\cap\mathbb{S}^{2N-1}$. Then (\ref{sub-gau-incre}) suggests that $\mathcal{R}(\bm{z},X)$ has sub-Gaussian increment with respect to the metric $d(\bm{z}_1,\bm{z}_2)=\|\bm{z}_1-\bm{z}_2\|_2\cdot\sqrt{\frac{60}{N-1}}$. By Dudley's entropy bound, we have
    \begin{align}
                \label{Dudley}
        \EE[\mathcal{R}(X)]&\leq 16\int_0^{+\infty}\sqrt{\log\left(\mathcal{N}\left(E\cap\mathbb{S}^{2N-1},d,\varepsilon\right)\right)}d\varepsilon\notag\\
        &\leq \frac{125}{\sqrt{N-1}}\int_0^{\operatorname{diam}(E\cap\mathbb{S}^{2N-1})}\sqrt{\log\left(\mathcal{N}\left(E\cap\mathbb{S}^{2N-1},\|\cdot\|_2,\varepsilon\right)\right)}d\varepsilon\notag\\
        &\leq\frac{125}{\sqrt{N-1}}\int_0^2 \sqrt{m\log\frac{3}{\varepsilon}}d\varepsilon\leq 285\sqrt{\frac{m}{N-1}},
    \end{align}
    where the covering number $\mathcal{N}\left(E\cap\mathbb{S}^{2N-1},\|\cdot\|_2,\varepsilon\right)$ is the number of $\varepsilon$ balls needed to cover $E\cap\mathbb{S}^{2N-1}$. Combining (\ref{eq:upperbounf:R(X)}),(\ref{eq:concentration for R(X)}),(\ref{Dudley}), we have
        \begin{equation}
        \label{eq: prob R(X)}
        \PP\left(\mathcal{R}(X)\geq t+285\sqrt{\frac{m}{{N-1}}}\right)\leq \exp\left(-\frac{(N-1)t^2}{30}\right).
    \end{equation}
    For $U=g(X)$, if $\mathcal{R}(X)\leq\varepsilon$, then we have for all $\bm{z}\in E\cap\mathbb S^{2N-1}$,  $-\varepsilon\leq \mathcal{R}(\bm{z},X)\leq \varepsilon$, which means that $\bm{z}^\top(I_{2N}-P_{S^U}(\bm{z}_{\bm\theta}))\bm{z}\in[\frac{1}{2}-\varepsilon,\frac{1}{2}+\varepsilon]$. As a result,
    \[(\frac{1}{2}-\varepsilon)\bm{y}^\top A(\bm\theta)^\top A(\bm\theta)\bm{y}\leq \bm{y}^\top A(\bm\theta)^\top (I_{2N}-P_{S^U}(\bm{z}_{\bm\theta}))A(\bm\theta)\bm{y}\leq (\frac{1}{2}+\varepsilon)\bm{y}^\top A(\bm\theta)^\top A(\bm\theta)\bm{y}, \]
    for all $\bm{y}\in\RR^{m}$, and this is equivalent to $(\frac{1}{2}-\varepsilon)Q(\bm\theta)\preceq F^U(\bm\theta)\preceq(\frac{1}{2}+\varepsilon)Q(\bm\theta)$. Then \thmref{5} is obtained by taking $t=\varepsilon-285\sqrt{\frac{m}{N-1}}$ in (\ref{eq: prob R(X)}).
\end{proof}
\begin{remark}
\label{remark:thm5}
   The problem can be framed as using a random projection operator $r(U) = I_{2N} - P_{S^U}(z_{\bm\theta})$ to approximate the identity operator within a fixed subspace $E$ of dimension $m$. In the case where $r(U) = U$, this concentration is a direct consequence of Dvoretzky’s Theorem. Following the framework in \cite{meckes2019random}, this concentration principle extends to operators that are Lipschitz continuous on the unitary group $\mathrm{U}(N)$ with respect to the Hilbert-Schmidt (Frobenius) norm. The resulting error scaling of $\sqrt{m/N}$ is a fundamental characteristic in random matrix theory, reflecting the uniform isometry property (or restricted isometry) when projecting high-dimensional measures onto a $m$-dimensional subspace, much like the scaling observed in the Johnson–Lindenstrauss Lemma. We explain this as follows:
\end{remark}
\begin{theorem}[\cite{Vershynin_2018}, Johnson–Lindenstrauss Lemma]
    Let $V \subseteq \mathbb{R}^d$ be a set of $n$ points. For any error parameter $0 < \epsilon < 1$, there exists a linear mapping $f: \mathbb{R}^d \to \mathbb{R}^k$ such that for all $u, v \in V$, the following inequality holds:$$(1 - \epsilon) \|u - v\|^2 \le \|f(u) - f(v)\|^2 \le (1 + \epsilon) \|u - v\|^2.$$ To satisfy this distance-preserving property, $k$ only needs to satisfy:$$k \ge O\left( \frac{\log n}{\epsilon^2} \right).$$
\end{theorem}
It was proven by using random projection and estimating the expected error $\varepsilon$. 
In our case, $V$ is in the unit sphere in the subspace $E$ and $f(\bm u)=2(I-P_{S^U}(z_{\bm\theta}))\bm u, k=2N$. Then \[\varepsilon\approx\sup_{\bm u,\bm v\in E\cap \mathbb S^{m-1}} \left|\frac{\|f(\bm u)-f(\bm v)\|_2}{\|\bm u-\bm v\|_2}-1\right|\leq\sup_{\bm u\in E\cap \mathbb S^{m-1}}\|(I-2P_{S^U}(z_{\bm\theta}))\bm u\|_2\ .\] 
The right term is the exact quantity we need to estimate. To control the norm uniformly in $E\cap \mathbb S^{m-1}$, usually $n=|V|\approx \varepsilon^{-m}$, so $\varepsilon=\Omega\left(\sqrt{\frac{\log n}{k}}\right)=\tilde \Omega\left(\sqrt{\frac{m}{N}}\right)$.

\section{Conclusion}\label{sec:conclusion}
This work shows an interesting relationship between the classical Fisher information matrix (CFIM) under random measurements and the quantum Fisher information matrix (QFIM). By studying real representations of these two kinds of information matrices, we find an elegant way to illustrate the transformation of the random CFIM between different bases. Thereafter, we rigorously derive the variance and estimate moments of the random CFIM by exploiting the symmetry of the Haar distribution on the unitary group. Moreover, we provide  matrix concentration analysis for the CFIM based on a well-developed technique that proves concentration on unitary groups. The key step is to identify the Lipschitz continuity of the CFIM with respect to its measurement basis. The lower bound of exponential moment and numerical experiments demonstrate that the upper bound in our derived concentration inequality is probably optimal in the exponent up to a constant.

As possible future directions, as this work only considers pure quantum states, it is natural to investigate whether there is a similar relationship for mixed quantum states (see e.g., Theorem 2.2 in \cite{Liu_2020} for definition of QFIM for mixed states). Another direction is to consider practical (easy to implement) unitary ensemble $\nu\subseteq \mathrm U(N)$ on a quantum computer that the average CFIM over $\nu$ serves as a good estimator for the QFIM. As Lemma \ref{thm: conditional expectation of Pk} suggests, when imposing certain conditions on the unitary $U$, the expectation becomes a $2$-moment of an operator $X\sim\mu_H$ on $\mathrm U(N-1)$ ($U=g(X)$, see the definition of $g$ in Lemma \ref{lem:keyLip}). Hence, we could possibly find a well-behaved $\nu$ based on unitary $k$-designs. {This has been substantiated by recent breakthrough in} \cite{zhou2025randomizedmeasurementsmultiparameterquantum}, {which proved that using unitary 3-designs, the average CFIM satisfies $\EE[F(\bm\theta)]\succeq\frac{1}{4}Q(\bm\theta)$. Concentration results for such measurement ensembles might be interesting. We leave this for future work.}  

\bibliographystyle{amsxport}
\bibliography{qfim}

\end{document}